\documentclass{lmcs}
\pdfoutput=1

\usepackage{lastpage}
\lmcsdoi{15}{4}{9}
\lmcsheading{}{\pageref{LastPage}}{}{}%
{Jun.~18,~2018}{Nov.~14,~2019}{}

\usepackage{xspace}
\usepackage{amssymb}
\usepackage{amsmath}
\usepackage{amsthm}
\usepackage{latexsym}
\usepackage{graphicx}
\usepackage{color}
\usepackage{dsfont}
\usepackage{epic,eepic}

\newcommand{\N}{{\rm I}\!{\rm N}}
\newcommand{\R}{{\rm I}\!{\rm R}}
\newcommand{\zug}[1]{\langle#1\rangle}
\newcommand{\stam}[1]{}
\newcommand{\ctls}{CTL$^\star$\xspace}

\newcommand{\lfl}{LFL\xspace}
\newcommand{\bfls}{BFL$^\star$\xspace}

\newcommand{\bflso}{BFL$^\star_1$\xspace}
\newcommand{\bflsok}{BFL$^\star_{1,k}$\xspace}
\newcommand{\bfl}{BFL\xspace}
\newcommand{\ebflso}{$\exists$BFL$^\star_1$\xspace}
\newcommand{\abflso}{$\forall$BFL$^\star_1$\xspace}
\newcommand{\lflo}{LFL$_1$\xspace}

\newcommand{\elflo}{$\exists$LFL$_1$\xspace}
\newcommand{\alflo}{$\forall$LFL$_1$\xspace}
\newcommand{\cbfls}{CBFL$^\star$\xspace}
\newcommand{\ecbflso}{$\exists$CBFL$^\star_1$\xspace}
\newcommand{\acbflso}{$\forall$CBFL$^\star_1$\xspace}
\newcommand{\cbflsi}{CBFL$^\star_i$\xspace}
\newcommand{\cbflsipo}{CBFL$^\star_{i+1}$\xspace}
\newcommand{\ecbflsipo}{$\exists$CBFL$^\star_{i+1}$\xspace}

\newcommand{\ecbflsi}{$\exists$CBFL$^\star_i$\xspace}
\newcommand{\acbflsi}{$\forall$CBFL$^\star_i$\xspace}
\newcommand{\cbflsz}{CBFL$^\star_0$\xspace}
\newcommand{\cbflso}{CBFL$^\star_1$\xspace}

\newcommand{\true}{{\bf true}}
\newcommand{\false}{{\bf false}}

\newcommand{\Af}{{\mathcal{A}}}
\newcommand{\Ef}{{\mathcal{E}}}

\newcommand{\lopen}[1]{(#1]} % chktex 9
\newcommand{\ropen}[1]{[#1)} % chktex 9


\begin{document}

\title[Flow Logic]{Flow Logic}

\author[O. Kupferman]{Orna Kupferman}
\address{School of Computer Science and Engineering, The Hebrew University, Israel}
\thanks{This research has received funding from the European Research Council under the EU's 7-th Framework Programme (FP7/2007--2013) / ERC grant agreement no 278410.}

\author[G. Vardi]{Gal Vardi}
\email{gal.vardi@mail.huji.ac.il}

\begin{abstract}
A flow network is a directed graph in which each edge has a capacity, bounding the amount of flow that can travel  through it. Flow networks have attracted a lot of research in computer science. Indeed, many questions in numerous application areas can be reduced to questions about flow networks. This includes direct applications, namely a search for a maximal flow in networks, as well as less direct applications, like maximal matching or optimal scheduling. Many of these applications would benefit from a framework in which one can  formally reason about properties of flow networks that go beyond their maximal flow.

We introduce \emph{Flow Logics}: modal logics that treat flow functions as explicit first-order objects and enable the specification of rich properties of flow networks. The syntax of our logic \bfls (Branching Flow Logic) is similar to the syntax of the temporal logic \ctls, except that atomic assertions may be \emph{flow propositions}, like $> \gamma$ or $\geq \gamma$, for $\gamma \in \N$, which refer to the value of the flow in a vertex, and that first-order quantification can be applied both to paths and to flow functions. For example, the \bfls formula $\Ef ((\geq 100) \wedge AG({\mathit{low}} \rightarrow (\leq 20)))$ states that there is a legal flow function in which the flow is above $100$ and in all paths, the amount of flow that travels through vertices with low security is at most $20$.

We present an exhaustive study of the theoretical and practical aspects of \bfls, as well as extensions and fragments of it. Our extensions include flow quantifications that range over non-integral flow functions or over maximal flow functions, path quantification that ranges over paths along which non-zero flow travels, past operators, and first-order quantification of flow values. We focus on the model-checking problem and show that it is PSPACE-complete, as it is for \ctls. Handling of flow quantifiers, however, increases the complexity in terms of the network to ${\rm P}^{\rm NP}$, even for the LFL and BFL fragments, which are the flow-counterparts of LTL and CTL\@. We are still able to point to a useful fragment of \bfls for which the model-checking problem can be solved in polynomial time.
Finally, we introduce and study the query-checking problem for \bfls, where under-specified \bfls formulas are used for network exploration.
\end{abstract}

\maketitle

\section{Introduction}%
\label{intro}
A \emph{flow network} is a directed graph in which each edge has a
capacity, bounding the amount of flow that can travel through it. The
amount of flow that enters a vertex equals the amount of flow that
leaves it, unless the vertex is a \emph{source}, which has only
outgoing flow, or a \emph{target}, which has only incoming flow.
The fundamental \emph{maximum-flow problem} gets as input a flow network and searches for a maximal flow
from the source to the target~\cite{CLR90,GTT89}.
The problem was first formulated and solved in the 1950's~\cite{FF56,FF62}. It has attracted much research on improved algorithms~\cite{EC72,Din70,GT88,Mad16} and applications~\cite{AMO93}.

The maximum-flow problem can be applied in many settings in which something
travels along a network. This covers numerous application domains,
including traffic in road or rail systems, fluids in pipes, currents
in an electrical circuit, packets in a communication network, and
many more~\cite{AMO93}.
Less obvious applications involve flow networks
that are constructed in order to model settings with an abstract network,
as in the case of scheduling with constraints~\cite{AMO93} or elimination in partially completed tournaments~\cite{Sch66}.
In addition, several classical graph-theory problems can be reduced to the
maximum-flow problem.  This includes the problem of finding a maximum
bipartite matching, minimum path cover, maximum edge-disjoint or
vertex-disjoint path, and many more~\cite{CLR90,AMO93}. Variants of the
maximum-flow problem can accommodate further settings, like circulation
problems, where there are no sink and target vertices, yet there is a
lower bound on the flow that needs to be traversed along each edge~\cite{Tar85}, networks with multiple source and target vertices, networks
with costs for unit flows, networks with multiple commodities, and more~\cite{EIS76}.

All the above applications reduce the problem at hand to the problem of finding a maximal flow in a network. Often, however, one would like to reason about properties of flow networks that go beyond their maximal flow. This is especially true when the vertices or edges of the network attain information to which the properties can refer. For example, the vertices of a network may be labeled by their security level, and we may want to check whether all legal flow functions are such that the flow in every low-security vertex is at most $20$, or check whether there is a flow function in which more than $100$ units of flow reach the target and still the flow in every low-security vertex is at most $20$. As another example, assume that each vertex in the network is labeled by the service provider that owns it, and we want to find a maximal flow under the constraint that flow travels through vertices owned by at most two providers.
%gal1: how to you model the providers example with flow logic?
%orna1 \bigvee_{p,q \in providers} \Ef^max \A^G p \/ A^G q
%gal2: I still don't understand
%orna2 we want all the reachable vertices to be labeled by only p or q (p and q are the guessed two providers).

The challenge of reasoning about properties of systems has been extensively studied in the context of formal verification. In particular, in \emph{temporal-logic model checking}~\cite{CE81,QS82}, we check whether a system has a desired property by translating the system into a labeled state-transition graph, translating the property into a temporal-logic formula, and deciding whether the graph satisfies the formula.
Model checking is one of the notable success stories of theoretical computer science, with exciting theoretical research that is being transformed into industrial applications~\cite{CGP99,CHVB18}.
By viewing networks as labeled state-transition graphs, we can use existing model-checking algorithms and tools in order to reason about the structural properties of networks. We can check, for example, that every path from the source to the target eventually visits a \textit{check-sum} vertex. Most interesting properties of flow networks, however, refer to flows and their values, and not just to the structural properties of the network. Traditional temporal logics do not support the specification and verification of such properties.

We introduce and study \emph{Flow Logics}: modal logics that treat flow functions as explicit first-order objects and enable the specification of rich properties of flow networks. The syntax of our logic \bfls (Branching Flow Logic) is similar to the syntax of the temporal logic \ctls, except that atomic assertions are built from both atomic propositions and \emph{flow propositions}, like $> \gamma$ or $\geq \gamma$, for $\gamma \in \N$, which refer to the value of the flow in a vertex, and that first-order quantification can be applied both to paths and to flow functions. Thus, in addition to the path quantifiers $A$ (``for all paths'') and $E$ (``there exists a path'') that range over paths, states formulas may contain the flow  quantifiers $\Af$ (``for all flow functions'') and $\Ef$ (``there exists a flow function''). For example, the \bfls formula $\Ef ((\geq 100) \wedge AG(\mathit{low} \rightarrow (\leq 20)))$ states the property discussed above, namely that there is a flow function in which the value of the flow is at least $100$, and in all paths, the value of flow
that travels through
in vertices with low security is at most~$20$.

We study the theoretical aspects of \bfls as well as extensions and fragments of it. We demonstrate their applications in reasoning about flow networks, and we examine the complexity of their model-checking problem.
Below we briefly survey  our results. (1) We show that while maximal flow can always be achieved by integral flows~\cite{FF56}, in the richer setting of flow logic, restricting attention to integral flows may change the satisfaction value of formulas. Accordingly, our semantics for \bfls considers two types of flow quantification: one over integral flows and another over \emph{non-integral} ones. (2) We prove that bisimulation~\cite{Mil71} is not a suitable equivalence relation for flow logics, which are sensitive to unwinding. We relate this to the usefulness of \emph{past operators} in flow logic, and we study additional aspects of the expressive power of \bfls. (3) We consider extensions of \bfls by path quantifiers that range over paths on which flow travels (rather than over all paths in the network), and by \emph{first-order quantification on flow values}. (4) We study the model-checking complexity of \bfls, its extensions, and some natural fragments. We show that algorithms for temporal-logic model-checking can be extended to handle flow logics, and that the complexity of the \bfls model-checking problem is PSPACE-complete. We study also the \emph{network complexity} of the problem, namely the complexity in terms of the network, assuming that the formula is fixed~\cite{LP85,KVW00}, and point to a fragment of \bfls for which the model-checking problem can be solved in polynomial time.

%of stam short

One of the concepts that has emerged in the context of formal
verification is that of \emph{model exploration}. The idea, as first
noted by Chan in~\cite{Cha00}, is that, in practice, model checking is
often used for understanding the system rather than for verifying its correctness.
Chan suggested to formalize model exploration by means of \emph{query
checking}.  The input to the query-checking problem is a system $S$ and
a query $\varphi$, where a query is a temporal-logic formula in which
some proposition is replaced by the place-holder ``?''.
A solution to the query is a propositional assertion that, when
replaces the place-holder, results in a formula that is satisfied in
$S$.  For example, solutions to the query $AG (? \rightarrow AX \mathit{grant})$ are assertions $\psi$ for which $S\models AG(\psi \rightarrow AX \mathit{grant})$, thus assertions that trigger a grant in all successive possible futures.
A query checker should return the strongest solutions to the query (strongest
in the sense that they are not implied by other
solutions). The work
of Chan was followed by further work on query checking, studying its
complexity, cases in which only a single strongest solution exists,
the case of multiple (possibly related) place-holders, queries on semantic graphs, understanding business processes through query checking and more~\cite{BG01,GC02,CG03,SV03,CGG07,GRN12,RDMMM14}.
%orna5 what about these?
%orna4 add recent references
%gal5 added three
In the context of flow networks, the notion of model explorations is of special interest, and there is on-going research on the development of automatic tools for reasoning about networks and flow-forwarding strategies (c.f.,~\cite{BDVV18,ZSTLLM10}).
%orna4 I believe that in the intro of BDVV there are more references
%gal5 what is BDVV18? I don't see it in ok.bib
%orna6: it was in our main dropbox directory. I changed the call.
%gal6: added one more reference
We develop a theory of query checking for \bfls.
We distinguish between \emph{propositional queries}, which, as in the case of \ctls, seek strongest propositional assertions that may replace a $?$ that serves as a state formula, and \emph{value queries}, which seek strongest values that may replace a $?$ that serves as a bound in a flow proposition. There, the strength of a solution depends on the tightness of the bound it imposes. We show that while in the propositional case there may be several partially ordered strongest solutions, in the value case the solutions are linearly ordered and thus there is at most one strongest solution. % chktex 40
 %and we can ask an algorithm to return the strongest one.
%orna5
%gal5 let's finish with the results first
%%orna6 I think we are done
%gal6: OK. We will finish this part this evening
The existence of a single strongest solution has some nice practical applications. In particular, we show that while for propositional queries, the problem of finding all strongest solutions is very complex, for value queries its complexity coincides with that of \bfls model checking.
%Also, if $\gamma \in \N$ is a strongest solution to a value query, then the set of solutions to the query is either $\N \cap [\gamma,\infty]$ or $\N \cap [0,\gamma]$. Moreover, we show that while finding all strongest solutions to a propositional query $\psi$ is a very complex task (involving $2^{2^{|AP|}}$ executions of a model checker for a formula of length $O(|\psi|+2^{|AP|})$), finding the strongest solution to a value query is PSPACE-complete, and its network complexity is in $\Delta_2^P$.

\subsection*{Related work}
There are three types of related works: (1) efforts to generalize the maximal-flow problem to richer settings, (2) extensions of temporal logics by new elements, in particular first-order quantification over new types, and (3) works on logical aspects of networks and their use in formal methods.  Below we briefly survey them and their relation to our work.

As discussed early in this section, numerous extensions to the classical maximal-flow problems have been considered. In particular, some works that add constraints on the maximal flow, like capacities on vertices, or lower bounds on the flow along edges. Closest to flow logics are works that refer to labeled flow networks. For example,~\cite{GCSR13} considers flow networks in which edges are labeled, and the problem of finding a maximal flow with a  minimum number of labels. Then, the maximal utilization problem of \emph{capacitated automata}~\cite{KT14} amounts to finding maximal flow in a labeled flow network where flow is constrained to travel only along paths that belong to a given regular language.
Our work suggests a formalism that embodies all these extensions, as well as a framework for formally reasoning about many more extensions and settings.

The competence of temporal-logic model checking initiated numerous extensions of temporal logics, aiming to capture richer settings. For example, \emph{real-time} temporal logics include clocks with a real-time domain~\cite{AH94a}, \emph{epistemic temporal logics} include knowledge operators~\cite{HM90}, and \emph{alternating temporal logics} include game modalities~\cite{AHK02}. Closest to our work is \emph{strategy logic}~\cite{CHP07}, where temporal logic is enriched by first-order quantification of strategies in a game. Beyond the theoretical interest in strategy logic, it was proven useful in synthesizing strategies in multi-agent systems and in the solution of rational synthesis~\cite{FKL10}.

Finally, network verification is an increasingly important topic in the context of protocol verification~\cite{KVM12}. Tools that allow verifying properties of network protocols have been developed~\cite{WBLS09,LBGJV15}. These tools support verification of network protocols in the design phase as well as runtime verification~\cite{KCZVMW13}. Some of these tools use a query language called \emph{Network Datalog} in order to specify network protocols~\cite{LCGGHMRRS06}. Verification of \emph{Software Defined Networks} has been studied widely, for example in~\cite{KZZCG13,KPCVK12,CVPKR12}. Verification of safety properties in networks with finite-state middleboxes was studied in~\cite{VAPRSSS16}. Network protocols describe forwarding policies for packets, and are thus related to specific flow functions. However, the way traffic is transmitted in these protocols does not correspond to the way flow  travels in a flow network. Thus, properties verified in this line of work are different from these we can reason about with flow logic.

%Due to the lack of space, some details and proofs are omitted, and can be found in the full version, in the authors' URLs.

\section{The Flow Logic \texorpdfstring{\bfls}{BFL*}}
A \emph{flow network} is $N=\zug{AP, V,E,c,\rho,s,T}$, where $AP$ is a set of atomic propositions, $V$ is a set of vertices, $s \in V$ is a source vertex,  $T \subseteq V \setminus \{s\}$ is a set of target vertices, $E \subseteq (V \setminus T) \times (V \setminus \{s\})$ is a set of directed edges,
%gal2: we allow self loops
%orna2: is there a problem with it?
$c:E \rightarrow \N$ is a capacity function, assigning to each edge an integral amount of flow that the edge can transfer, and $\rho:V \rightarrow 2^{AP}$ assigns each vertex $v \in V$ to the set of atomic propositions that are valid in $v$. Note that no edge enters the source vertex or leaves a target vertex. We assume that all vertices $t \in T$ are reachable from $s$ and that each vertex has at least one target vertex reachable from it. For a vertex $u \in V$, let $E^u$ and $E_u$ be the sets of incoming and outgoing edges to and from $u$, respectively. That is, $E^u=(V \times \{u\}) \cap E$ and $E_u=(\{u\} \times V) \cap E$.

A \emph{flow} is a function $f:E \rightarrow \N$  that describes how flow is directed in $N$. The capacity of an edge bounds the flow in it, thus for every edge $e \in E$, we have $f(e) \leq c(e)$. All incoming flow must exit a vertex, thus for every vertex $v \in V \setminus (\{s\} \cup T)$, we have $\sum_{e \in E^v}f(e)=\sum_{e \in E_v}f(e)$.
We extend $f$ to vertices and use $f(v)$ to denote the flow that travels through $v$. Thus, for $v \in V \setminus (\{s\} \cup T)$, we  define $f(v)=\sum_{e \in E^v}f(e)=\sum_{e \in E_v}f(e)$, for the source vertex $s$, we define $f(s)= \sum_{e \in E_s}f(e)$, and for a target vertex $t \in T$, we define $f(t)= \sum_{e \in E^t}f(e)$. %We also define $f(T)= \sum_{t \in T}f(t)$ as the flow that reaches $T$.
Note that the preservation of flow in the internal vertices guarantees that $f(s)=\sum_{t \in T}f(t)$, which is the amount of flow that travels from $s$ to all the target vertices together. We say that a flow function $f$ is \emph{maximal} if for every flow function $f'$, we have $f'(s) \leq f(s)$. A maximal flow function can be found in polynomial time~\cite{FF62}. The maximal flow for $N$ is then $f(s)$ for some maximal flow function $f$.

The logic \bfls is a \emph{Branching Flow Logic} that can specify properties of networks and flows in them.
As in \ctls, there are two types of formulas in \bfls:
\emph{state formulas}, which describe vertices in a network, and
\emph{path formulas}, which describe paths. In addition to the operators in \ctls, the logic \bfls has \emph{flow propositions}, with which one can specify the flow in vertices, and \emph{flow quantifiers}, with which one can quantify flow functions universally or existentially. When flow is not quantified, satisfaction is defined with respect to both a network and a flow function.
Formally, given a set $AP$ of atomic propositions, a \bfls state formula is one of the following:
\begin{enumerate}[label={\textbf{(S\arabic*):}}]
\item
  An atomic proposition $p\in AP$.
\item
  A flow proposition $> \gamma$ or $\geq \gamma$, for an integer $\gamma \in \N$.
\item
  $\neg\varphi_1$ or $\varphi_1\vee\varphi_2$, for \bfls state  formulas $\varphi_1$
  and $\varphi_2$.
\item
  $A\psi$, for a \bfls path formula $\psi$ (and $A$ is termed a \emph{path quantifier}).
  \item
  $\Af \varphi$, for a \bfls state formula $\varphi$ (and $\Af$ is termed a \emph{flow quantifier}).
\end{enumerate}
A \bfls path formula is one of the following:
\begin{enumerate}[label={\textbf{(P\arabic*):}}]
\item
  A \bfls state formula.
\item
  $\neg\psi_1$ or $\psi_1\vee\psi_2$, for \bfls path  formulas $\psi_1$ and $\psi_2$.
\item
  $X\psi_1$ or $\psi_1U\psi_2$, for \bfls path  formulas $\psi_1$ and $\psi_2$.
\end{enumerate}

\noindent
We say that a \bfls formula $\varphi$ is \emph{closed} if all flow propositions appear in the scope of a flow quantifier. The logic \bfls consists of the set of closed \bfls state formulas.
We refer to state formula of the form $\Af \varphi$ as a \emph{flow state formula}.

The semantics of \bfls is defined with respect to vertices in a
flow network.
Before we define the semantics, we need some more definitions and notations. Let  $N=\zug{AP, V,E,c,\rho,s,T}$.
For two vertices $u$ and $w$ in $V$, a finite sequence $\pi=v_0,v_1,\ldots,v_k \in V^*$ of vertices is a \emph{$(u,w)$-path} in $N$ if $v_0=u$, $v_k=w$, and $\zug{v_i,v_{i+1}} \in E$ for all $0 \leq i < k$. If $w \in T$, then $\pi$ is a \emph{target $u$-path}.

State formulas are interpreted with respect to a vertex $v$ in $N$ and a flow function $f: E \rightarrow \N$. When the formula is closed, satisfaction is independent of the function $f$ and we omit it.
We use
$v,f \models \varphi$ to indicate that the vertex $v$ satisfies the state formula $\varphi$ when the flow function is $f$. The relation $\models$ is defined inductively as follows.

\begin{enumerate}[leftmargin=13mm,align=left]
\item[\textbf{(S1):}]
For an atomic proposition $p\in AP$, we have that $v,f\models p$ iff $p\in\rho(v)$.
\item[\textbf{(S2):}]
For $\gamma\in\N$, we have $v,f\models > \gamma$ iff $f(v) > \gamma$ and $v,f\models \geq \gamma$ iff $f(v) \geq \gamma$.
\item[\textbf{(S3a):}]
$v,f\models\neg\varphi_1$ iff $v,f\not\models\varphi_1$.
\item[\textbf{(S3b):}]
  $v,f\models\varphi_1\vee\varphi_2$ iff $v,f\models\varphi_1$ or $v,f\models\varphi_2$.
  \item[\textbf{(S4):}]
$v,f \models A\psi$ iff for all target $v$-paths $\pi$, we have that $\pi,f \models \psi$.
\item[\textbf{(S5):}]
$v,f\models \Af\varphi$ iff for all flow functions $f'$, we have $v,f' \models \varphi$.
\end{enumerate}

\noindent
Path formulas are interpreted with respect to a finite path $\pi$ in $N$ and a flow function $f:E \rightarrow \N$. We use
$\pi,f \models \varphi$ to indicate that the path $\pi$ satisfies the path-flow formula $\psi$ when the flow function is $f$. The relation $\models$ is defined inductively as follows.  Let $\pi=v_0,v_1,\ldots,v_k$. For $0 \leq i \leq k$, we use $\pi^i$ to denote the suffix of $\pi$ that starts at $v_i$, thus $\pi^i=v_i,v_{i+1},\ldots, v_k$.
\begin{enumerate}[leftmargin=13mm,align=left]
\item[\textbf{(P1):}]
   For a state formula $\varphi$, we have that $\pi,f\models \varphi$ iff $v_0,f \models \varphi$.
\item[\textbf{(P2a):}]
  $\pi,f \models\neg\psi$ iff $\pi,f \not\models\psi$.
\item[\textbf{(P2b):}]
  $\pi,f \models\psi_1\vee\psi_2$ iff $\pi,f \models\psi_1$ or
  $\pi,f \models\psi_2$.
\item[\textbf{(P3a):}]
  $\pi,f \models X\psi_1$ iff $k > 0$ and $\pi^1,f \models \psi_1$.
  \item[\textbf{(P3b):}]
  $\pi,f \models \psi_1 U \psi_2$ iff there is $j \leq k$ such that $\pi^j,f \models \psi_2$, and for all $0 \leq i < j$, we have $\pi^i,f \models \psi_1$
\end{enumerate}

\noindent
For a network $N$ and a closed \bfls formula $\varphi$, we say that $N$ satisfies $\varphi$, denoted $N \models \varphi$, iff $s \models \varphi$ (note that since $\varphi$ is closed, we do not specify a flow function).

Additional Boolean connectives and modal operators are defined from
$\neg$, $\vee$, $X$, and $U$ in the usual manner;
in particular, $F\psi\,=\,\true U\psi$ and $G\psi\,=\,\neg F\neg\psi$.
We also define dual and abbreviated flow propositions: $(<\gamma) = \neg (\geq \gamma)$, $(\leq \gamma) = \neg (> \gamma)$, and $(\gamma) = (=\gamma) = (\leq \gamma) \wedge (\geq \gamma)$, a dual path quantifier: $E\psi=\neg A\neg \psi$, and a dual flow quantifier: $\Ef \varphi = \neg \Af \neg \varphi$.

\begin{exa}
Consider a network $N$ in which target vertices are labeled by an atomic proposition $\mathit{target}$, and low-security vertices are labeled $\mathit{red}$. The \bfls formula $\Ef EF(\mathit{target} \wedge (=20))$ states that there is a flow in which $20$ units reach a target vertex,
and the \bfls formula $\Af ((\geq 20) \rightarrow AX (\geq 4))$ states that in all flow functions in which the flow at the source is at least $20$, all the successors must have flow of at least $4$. Finally, $\Ef ((\geq 100) \wedge AG(\mathit{red} \rightarrow (\leq 20)))$ states that there is a flow of at least $100$ in which the flow in every low-security vertex is at most $20$, whereas $\Af ((> 200) \rightarrow EF(\mathit{red} \wedge (> 20)))$ states that when the flow is above $200$, then there must exist a low-security vertex in which the flow is above $20$.
%
%space
%As an example to a \bfls formula with nesting of flow quantifiers, consider the formula $\Ef AG[(\mathit{green} \wedge \Ef > 5) \rightarrow >5]$, stating that in every high-security vertex in which a flow of more than $5$ is possible, the flow should be greater than 5. As another example,
As an example to a \bfls formula with an alternating nesting of flow quantifiers, consider the formula $\Ef AG( <10 \rightarrow \Af < 15)$, stating that there is a flow such that wherever the flow is below $10$, then in every flow it would be below $15$.
\qed%
\end{exa}

\begin{rem}%
\label{fin paths}
Note that while the semantics of \ctls and LTL is defined with respect to infinite trees and paths, path quantification in \bfls ranges over finite paths.  We are still going to use techniques and results known for \ctls and LTL in our study. Indeed, for upper bounds, the transition to finite computations only makes the setting simpler. Also, lower-bound proofs for \ctls and LTL are based on an encoding of finite runs of Turing machines, and apply also to finite paths.

Specifying finite paths, we have a choice between
a weak and a strong semantics for the $X$ operator. In the weak semantics, the last vertex in a path satisfies $X \psi$, for all $\psi$. In particular, it is the only vertex that satisfies $X\false$. In the strong semantics, the last vertex does not satisfy $X \psi$, for all $\psi$. In particular, it does not satisfy $X\true$.
%gal2: Did you mean that in strong semantics the last state does not satisfy X\true?
We use the strong semantics.
\end{rem}

\section{Properties of \texorpdfstring{\bfls}{BFL*}}%
\label{section properties}

\subsection{Integral vs.\ non-integral flow functions}%
\label{reals}

Our semantics of \bfls considers \emph{integral flow functions}: vertices receive integral incoming flow and partition it to integral flows in the outgoing edges. Integral-flow functions arise naturally in settings in which the objects we transfer along the network cannot be partitioned into fractions, as is the case with cars, packets, and more.
Sometimes, however, as in the case of liquids, flow can be partitioned arbitrarily. In the traditional maximum-flow problem, it is well known that the maximum flow can be achieved by integral flows~\cite{FF56}. We show that, interestingly, in the richer setting of flow logic, restricting attention to integral flows may change the satisfaction value of formulas.

\begin{prop}%
\label{allow r}
Allowing the quantified flow functions in \bfls to get values in $\R$ changes its semantics.
\end{prop}

\noindent
\begin{minipage}{0.70\textwidth}
\begin{proof}
Consider the network on the right. The \bfls formula $\varphi=\Ef (1\wedge AX (> 0))$ states that there is a flow function in which the flow that leaves the source is $1$ and the flow of both its successors is strictly positive. It is easy to see that while no integral flow function satisfies the requirement in $\varphi$, a flow function in which $1$ unit of flow in $s$ is partitioned between $u$ and $v$ does satisfy it.
\end{proof}
\end{minipage} \ \hspace{.2in}
\begin{minipage}{0.25\textwidth}
\includegraphics[scale=0.3]{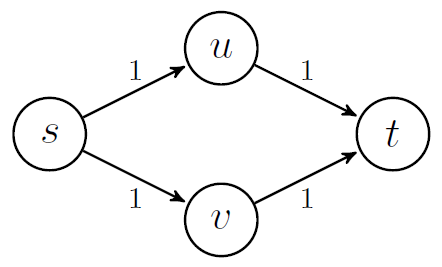}
\end{minipage}

\stam{
\begin{proof}
Consider the network $N$ appearing in Figure~\ref{real fig}. The \bfls formula $\varphi=\Ef (1\wedge AX (> 0))$ states that there is a flow function in which the flow that leaves the source is $1$ and the flow of both its successors is strictly positive. It is easy to see that while no integral flow function satisfies the requirement in $\varphi$, a flow function in which $1$ unit of flow in $s$ is partitioned between $u$ and $v$ does satisfy it.
\begin{figure}[ht]
\begin{center}
%%% \vspace{-5mm}
\includegraphics[scale=0.32]{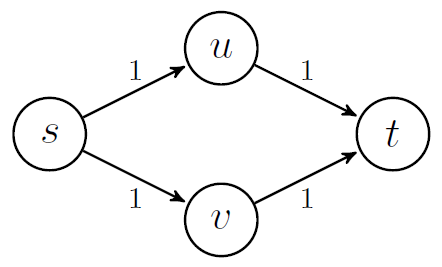}
\normalsize
\caption{The flow network $N$}%
\label{real fig}
\end{center}
\end{figure}
\end{proof}
}

\vspace{2mm}
Proposition~\ref{allow r} suggests that quantification of flow functions that allow non-integral flows may be of interest.  In Section~\ref{with ni flow} we discuss such an extension.

\subsection{Sensitivity to unwinding}%
\label{unwinding}
For a network $N=\zug{AP, V,E,c,\rho,s,T}$, let $N_t$ be the unwinding of $N$ into a tree.  Formally, $N_t=\zug{AP,V',E',\rho',s,T'}$, where $V' \subseteq V^*$ is the smallest set such that $s \in V'$, and for all $w \cdot v \in V'$ with $w \in V^*$ and $v \in V \setminus T$, and all $u \in V$ such that $E(v,u)$, we have that $w \cdot v \cdot u \in V'$, with $\rho'(w \cdot v \cdot u)=\rho(u)$. Also, $\zug{w \cdot v,w \cdot v \cdot u} \in E'$, with $c'(\zug{w \cdot v,w \cdot v \cdot u})=c(\zug{v,u})$. Finally, $T'=V' \cap (V^*\cdot T)$. Note that $N_t$ may be infinite. Indeed, a cycle in $N$ induces infinitely many vertices in $N_t$.

 The temporal logic \ctls is insensitive to unwinding. Indeed, $N$ and $N_t$ are bisimilar, and  for every \ctls formula $\varphi$, we have $N \models \varphi$ iff $N_t \models \varphi$~\cite{Mil71}.  We show that this is not the case for \bfls.

\begin{prop}
The value of the maximal flow is sensitive to unwinding.
\end{prop}
\begin{proof}
Consider the network $N$ appearing in Figure~\ref{n and nt} (with the target $T=\{t\}$), and its unwinding $N_t$ which appears in its right. It is easy to see that the value of the maximal flow in $N$ is $7$, and the value of the maximal flow from $s$ to $T'$ in $N_t$ is $8$.
\end{proof}

\begin{figure}[ht]
\begin{center}
%%% \vspace{-5mm}
\includegraphics[scale=0.35]{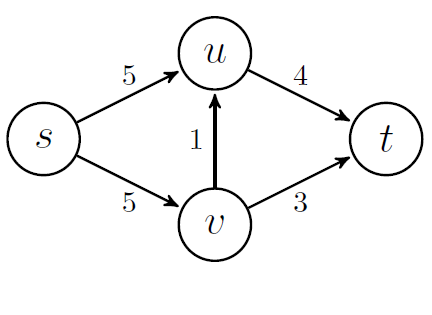}
\hspace{0.5cm}
%%% \vspace{1cm}
\includegraphics[scale=0.35]{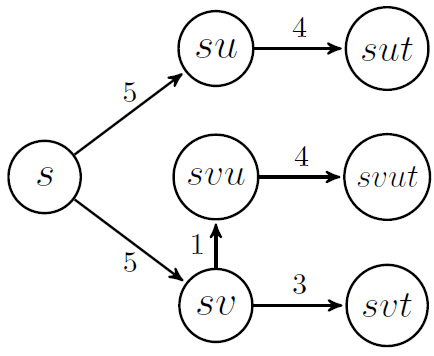}
\normalsize
\caption{The flow network $N$ and its unwinding $N_t$}%
\label{n and nt}
\end{center}
\end{figure}

\begin{cor}
The logic \bfls is sensitive to unwinding.
\end{cor}

\stam{
\begin{prop}
The logic \bfls is sensitive to unwinding.
\end{prop}

\begin{proof}
Consider the network $N$ appearing in Figure~\ref{n and nt}, and the \bfls formula $\varphi = \Af ((AX (\geq 2)) \rightarrow AXAX (\geq 2))$. The formula states that in all flow functions, if the flow in the successors of $s$ is at least $2$, then so is the flow in all their successors.
%gal2: Do you mean: $\varphi = \Af ((AX (\geq 2)) \rightarrow AXX (\geq 2))$
It is easy to see that $N \models \varphi$. Indeed, if the flow in both $u$ and $v$ is at least $2$, then so is the flow in $u$ and $t$. The unwinding $N_t$ of $N$, which appears in its right, does not satisfy $\varphi$. Indeed, there are flow functions in which the flow in $s \cdot u$ and $s \cdot v$ is at least $2$ and still the flow in vertex $s \cdot v \cdot u$ is $1$ or $0$.
\end{proof}
} %stam

The sensitivity of \bfls to unwinding suggests that extending \bfls with past operators can increase its expressive power. In Section~\ref{with past}, we discuss such an extension.

\stam{
\begin{figure}[ht]
\begin{minipage}[b]{0.55\linewidth}
\begin{center}
\includegraphics[scale=0.3]{before_unwinding.png}
\hspace{0.3cm}
\includegraphics[scale=0.3]{after_unwinding.png}
\normalsize
\caption{The flow network $N$ and its unwinding $N_t$.}%
\label{n and nt}
\end{center}
\end{minipage}
\quad
\begin{minipage}[b]{0.4\linewidth}
\begin{center}
%% \vspace{-10mm}
\includegraphics[scale=0.3]{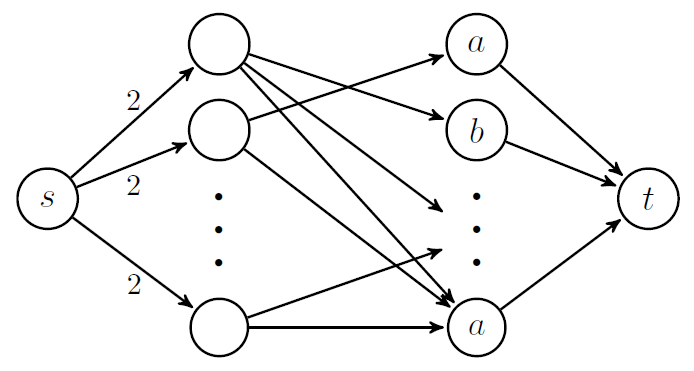}
\caption{Assigning workers to jobs.}%
\label{wt}
\end{center}
\end{minipage}
\end{figure}
}

\section{Extensions and Fragments of \texorpdfstring{\bfls}{BFL*}}%
\label{exts}
%gal3: why do we have a different subsection for every extension, but a single subsection for all the fragments together?
%orna3: since we don't have much to say about each fragment
In this section we discuss useful extensions and variants of \bfls, as well as fragments of it. As we shall show in the sequel, while the extensions come with no computational price, their model checking requires additional  techniques.
%gal2: We don't show for any extension a complexity of more than PSPACE

\subsection{Positive path quantification}%
\label{ppq}
Consider a network $N=\zug{AP, V,E,c,\rho,s,T}$ and a flow function $f: E \rightarrow \N$. We say that a path
$\pi=v_0,v_1,\ldots,v_k$ is positive if the flow along all the edges in $\pi$ is positive. Formally, $f(v_i,v_{i+1})>0$, for all  $0 \leq i < k$. Note that it may be that $f(v_i)>0$ for all $0 \leq i < k$ and still $\pi$ is not positive. It is sometimes desirable to restrict the range of path quantification to paths along which flow travels. This is the task of the \emph{positive path quantifier} $A^+$, with the following semantics.
\begin{itemize}
\item
$v,f \models A^+\psi$ iff for all positive target $v$-paths $\pi$, we have that $\pi,f \models \psi$.
%\item
%$v,f \models E^+\psi$ iff there exists a positive target $v$-path $\pi$ such that $\pi,f \models \psi$.
\end{itemize}

Dually, $E^+\psi = \neg A^+ \neg \psi$.

\begin{exa}%
\label{tw example}
%gal2: tasks --> jobs. T --> J.
Let $W$ be a set of workers and $J$ be a set of jobs. Each worker $w \in W$ can be assigned to perform jobs from a subset $J_w \subseteq J$. It is required to perform all jobs by assigning exactly one worker to each job and at most two jobs to each worker. This problem can be solved using the flow network appearing in Figure~\ref{wt}. The vertices on the left column correspond to the workers, these on the right column correspond to the jobs, and there is an edge between a worker $w$ and a job $j$ iff $j \in J_w$. The capacity of all edges is $1$, except for these from the source $s$ to the worker vertices, which have capacity $2$. It is easy to see that a flow of $k$ units in the network described in Figure~\ref{wt} corresponds to a legal assignment in which $k$ jobs are performed. Assume now that some jobs should be processed in Location $a$ and the others in Location $b$. A worker can process jobs only in a single location, $a$ or $b$. By labeling the job-vertices by their location, the existence of a legal assignment in which $k$ jobs are processed can be expressed by the
\bfls formula
\[
    \Ef (k \wedge AX(A^+Xa \vee  A^+Xb)),
\]
which uses positive path quantification.
%space
%following \bfls formula, which uses positive path quantification.
%$$\Ef (k \wedge AX(A^+Xa \vee  A^+Xb)).$$
\qed%
\end{exa}

\begin{figure}[ht]
\begin{center}
\includegraphics[scale=0.32]{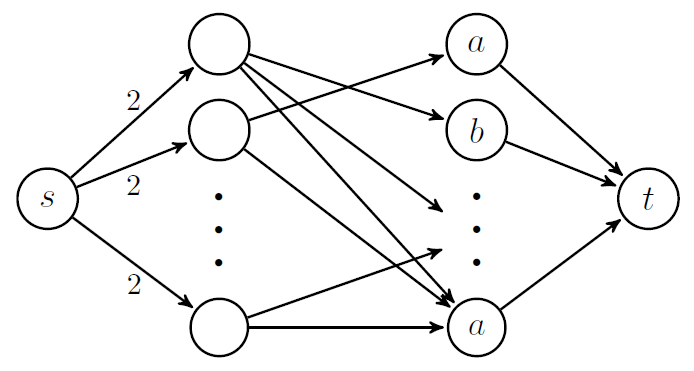}
\caption{Assigning workers to jobs.}%
\label{wt}
\end{center}
\end{figure}

\stam{
Moreover, assume now that the partition of the jobs between the locations is not given. The existence of a partition with which the required assignment exists and handles $k$ jobs amount to satisfaction of the \bfls formula
$\Ef (k \wedge A^+XEX0)$
in the network appearing in Figure~\ref{wt no locations}.
%gal2: but this problem is trivial. For example, we can label each job by a and then any solution is legal.
%orna2: we can fix it (say by requiring at least one worker in a and at least one in b), but let's remove this example, it is ``too NP'', and we save a figure :)
}

\subsection{Maximal flow quantification}%
\label{mfq}

 It is sometimes desirable to restrict the range of flow quantification to maximal flow functions. This is the task of the \emph{maximal-flow quantifier} $\Af^\mathit{max}$, with the following semantics.
\begin{itemize}
\item
$v,f \models\Af^\mathit{max}\varphi$ iff for all maximal-flow functions $f'$, we have that $v,f' \models \varphi$.
%\item
%$v,f \models \Ef^\mathit{max}\varphi$ iff there exists a maximal-flow function $f'$ such that $\pi,f' \models \varphi$.
\end{itemize}

\noindent
Dually, $\Ef^\mathit{max}\varphi = \neg \Af^\mathit{max} \neg \varphi$.

In a similar manner, it is sometimes helpful to relate to the maximal flow in the network. The max-flow constant $\gamma_\mathit{max} \in N$ maintains the value of the maximal flow from $s$ to $T$. We also allow arithmetic operations on $\gamma_\mathit{max}$.

\begin{exa}%
\label{tw example2}
Recall the job-assignment problem from Example~\ref{tw example}. The
\bfls
formula
\[\Ef^\mathit{max} (AX(A^+Xa \vee  A^+Xb))\] states that the requirements about the locations do not reduce the number of jobs assigned without this requirement. Then, the formula
\[\Ef ((\geq \gamma_\mathit{max} - 4) \wedge AX(A^+Xa \vee  A^+Xb))\] states that the requirements about the locations may reduce the number of jobs performed by at most~$4$.
\qed%
\end{exa}

\subsection{Non-integral flow quantification}%
\label{with ni flow}
As discussed in Section~\ref{reals}, letting flow quantification range over non-integral flow functions may change the satisfaction value of a \bfls formula.

Such a quantification is sometimes desirable, and we extend \bfls with a \emph{non-integral flow quantifier} $\Af^{\R}$, with the following semantics.
\begin{itemize}
\item
$v,f \models\Af^{\R}\varphi$ iff for all flow functions $f':E \rightarrow \R^+$, we have that $v,f' \models \varphi$.
%\item
%$v,f \models \Ef^{\N}\varphi$ iff there exists an integral flow function $f'$ such that $\pi,f' \models \varphi$.
\end{itemize}

Dually, $\Ef^{\R}\varphi = \neg \Af^{\R} \neg \varphi$.

\subsection{Past operators}%
\label{with past}

As discussed in Section~\ref{unwinding}, while temporal logics are insensitive to unwinding, this is not the case for \bfls. Intuitively, this follows from the fact that the flow in a vertex depends on the flow it gets from all its predecessors.
This dependency suggests that an explicit reference to predecessors is useful, and motivates the extension of \bfls by past operators.

Adding past to a branching logic, one can choose between a linear-past semantics ---~one in which past is unique (technically, the semantics is with respect to an unwinding of the network), and a branching-past semantics --- one in which all the possible behaviors that lead to present are taken into an account (technically, the semantics is dual to that of future operators, and is defined with respect to the network)~\cite{KPV12}. For flow logics, the branching-past approach is the suitable one, and is defined as follows.

For a path $\pi=v_0,v_1,\ldots,v_k \in V^*$, a vertex $v \in V$, and index $0 \leq i \leq k$, we say that $\pi$ is a source-target $(v,i)$-path if $v_0=s$, $v_i=v$, and $v_k \in T$. We add to \bfls two past modal operators, $Y$ (``Yesterday'') and $S$ (``Since''), and adjust the semantics as follows. Defining the semantics of logics that refer to the past, the semantics of path formulas is defined with respect to a path and an index in it. We use $\pi,i,f \models \psi$ to indicate that the path $\pi$ satisfies the path formula $\psi$ from position $i$ when the flow function is $f$.
%
%space
For state formulas, we adjust the semantics as follows.
\begin{enumerate}[leftmargin=12mm]
\item[\textbf{(S4):}]
$v,f \models A\psi$ iff for all $i$ and for all source-target $(v,i)$-paths $\pi$, we have that $\pi,i,f \models \psi$.
\end{enumerate}

\noindent
Then, for path formulas, we have the following (the adjustment to refer to the index $i$ in all other modalities is similar).
\begin{itemize}
\item
  $\pi,i,f \models Y\psi_1$ iff $i > 0$ and $\pi,i-1,f \models \psi_1$.
  \item
  $\pi,i,f \models \psi_1 S \psi_2$ iff there is $0 \leq j <i$ such that $\pi,j,f \models \psi_2$, and for all $j+1 \leq l \leq i$, we have $\pi,l,f \models \psi_1$.
 \end{itemize}

\begin{exa}
Recall the job-assignment problem from Example~\ref{tw example}. Assume that some of the workers have cars, which is indicated by an atomic proposition $\mathit{car}$ that may label worker vertices. Also, say that a job is a
\emph{transit job} if it can be assigned only to workers with cars. Assume that
we want to apply the restriction about the single location only to non-transit jobs. Thus, if all~the predecessors of a job-vertex are labeled by \textit{car}, then this job can be served in either locations. Using past operators, we can specify this property by the formula
\[\Ef (k \wedge AX(A^+X(a \vee AY \mathit{car}) \vee  A^+X(b \vee AY \mathit{car}))).\]
\qed%
\end{exa}

\subsection{First-Order quantification on flow values}%
\label{foq}
The flow propositions in \bfls include constants. This makes it impossible to relate the flow in different vertices other than specifying all possible constants that satisfy the relation. In \emph{\bfls with quantified flow values} we add flow variables $X=\{x_1,\ldots,x_n\}$ that can be quantified universally or existentially and specify such relations conveniently. We also allow the logic to apply arithmetic operations on the values of variables in $X$.

For a set of arithmetic operators $O$ (that is, $O$ may include $+, *$, etc.), let BFL$^\star(O)$ be \bfls in which Rule S2 is extended to allow expressions with variables in $X$ constructed by operators in $O$, and we also allow quantification on the variables in $X$. Formally, we have the following:
\begin{enumerate}[leftmargin=11mm]
\item[\textbf{(S2)}] A flow proposition $>g(x_1,\ldots,x_k)$ or $\geq~g(x_1,\ldots,x_k)$, where $x_1,\ldots,x_k$ are variables in $X$ and $g$ is an expression obtained from $x_1,\ldots,x_k$ by applying operators in $O$, possibly using constants in $\N$. We assume that $g:\N^k \rightarrow \N$. That is, $g$ leaves us in the domain $\N$.
\item[\textbf{(S6)}] $\forall x \varphi$, for $x \in X$ and a BFL$^\star(O)$ formula $\varphi$ in which $x$ is free.
\end{enumerate}

\noindent
For a BFL$^\star(O)$ formula $\varphi$ in which $x$ is a free variable, and a constant $\gamma\in \N$, let $\varphi[x \leftarrow \gamma]$ be the formula obtained by assigning $\gamma$ to $x$ and replacing expressions by their evaluation.
Then, $v,f \models \forall x \varphi'$ iff for all $\gamma \in \N$, we have that $v,f \models \varphi'[x \leftarrow \gamma]$.

\begin{exa}
The logic BFL$^\star(\emptyset)$ includes the formula
%space
$\Ef AG \forall x ((\mathit{split} \wedge x \: \wedge >0) \rightarrow EX (>0 \:\wedge <x))$,
stating that there is a flow in which all vertices that are labeled \textit{split} and with a positive flow $x$ have a successor in which the flow is positive but strictly smaller than $x$. Then,  BFL$^\star( {\rm div})$ includes the formula
$\Ef AG \forall x (x \rightarrow EX (\geq x \ {\rm div}\   2))$,
stating that there is a flow in which all vertices have a successor that has at least half of their flow.

Finally, BFL$^\star(+)$  includes the formula
$\exists x \exists y \Ef^\mathit{max} AG(\neg (\mathit{source} \vee \mathit{target}) \rightarrow x \vee y \vee (x+y))$,
stating that there are values $x$ and $y$, such that it is possible to attain the maximal flow by assigning to all vertices, except maybe source and target vertices,  values in $\{x,y,x+y\}$.
\qed%
\end{exa}

%of stam

%% \vspace{-4mm}
\subsection{Fragments of \texorpdfstring{\bfls}{BFL*}}%
\label{frags}
For the temporal logic \ctls, researchers have studied several fragments, most notably LTL and CTL\@. In this section we define interesting fragments of \bfls.

\vspace{3mm}
\noindent
{\bf  Flow-\ctls and Flow-LTL.}
The logics \emph{Flow-\ctls} and \emph{Flow-LTL} are extensions of \ctls and LTL in which atomic state formulas may be, in addition to $AP$s, also the flow propositions $> \gamma$ or $\geq \gamma$, for an integer $\gamma \in \N$. Thus, no quantification on flow is allowed, but atomic formulas may refer to flow.
The semantics of Flow-\ctls is defined with respect to a network and a flow function, and that of Flow-LTL is defined with respect to a path in a network and a flow function.

\vspace{3mm}
\noindent
{\bf Linear Flow Logic.}
The logic \lfl is the fragment of \bfls in which only one external universal path quantification is allowed. Thus, an \lfl formula is a \bfls formula of the form $A \psi$, where $\psi$ is generated without Rule S4.

Note that while the temporal logic LTL is a ``pure linear'' logic, in the sense that satisfaction of an LTL formula in a computation of a system is independent of the structure of the system, the semantics of LFL mixes linear and branching semantics. Indeed, while all the paths in $N$ have to satisfy $\psi$, the context of the system is important. To see this, consider the LFL formula $\varphi=A \Af ((\geq 10) \rightarrow X (\geq 4))$. The formula states that in all paths, all flow functions in which the flow at the first vertex in the path is at least $10$, are such that the flow at the second vertex in the path is at least $4$. In order to evaluate the path formula $\Af ((\geq 10) \rightarrow X (\geq 4))$ in a path $\pi$ of a network $N$ we need to know the capacity of all the edges from the source of $N$, and not only the capacity of the first edge in $\pi$. For example, $\varphi$ is satisfied in networks in which the source $s$ has two successors, each connected to $s$ by an edge with capacity $4$, $5$, or $6$. Consider now the  LFL formula $\varphi'=A \Ef (10 \wedge X (\geq 4))$. Note that $\varphi'$ is not equal to the \bfls formula $\theta=\Ef (10 \wedge AX (\geq 4))$. Indeed, in the latter, the same flow function should satisfy the path formula $X (\geq 4)$ in all paths.
%So, for example, $\theta$ cannot be satisfied in networks in which the source has three successors, whereas $\varphi'$ may be satisfied in such networks.
%gal2: unless the successors are neighbours of each other...
%orna2 we dn't claim that it is satisfied in all these networks. I would make it clearer.

\vspace{3mm}
\noindent
{\bf No nesting of flow quantifiers.}
%The family \emph{\bflsi}, for $i \geq 0$. The logic \bflsi is the fragment of \bfls in which nesting of at most $i$ flow quantifiers is allowed. In particular,
The logic \bflso contains formulas that are Boolean combinations of formulas of the form $\Ef\varphi$ and $\Af \varphi$, for a Flow-\ctls formula $\varphi$. Of special interest are the following  fragments of \bflso:
\begin{itemize}
\item
\ebflso and \abflso, where formulas are of the form  $\Ef\varphi$ and $\Af \varphi$, respectively, for a Flow-\ctls formula $\varphi$.
\item
\elflo and \alflo, where formulas are of the form $\Ef A\psi$ and $\Af  A\psi$, respectively, for a Flow-LTL formula $\psi$, and \lflo, where a formula is a Boolean combination of \elflo and \alflo formulas.
\end{itemize}

\vspace{3mm}
\noindent
{\bf Conjunctive-\bfls.}
The fragment Conjunctive-\bfls (\cbfls, for short) contains \bfls formulas whose flow state sub-formulas restrict the quantified flow in a conjunctive way. That is, when we ``prune'' a \cbfls formula into requirements on the network, atomic flow propositions are only conjunctively related. This would have a computational significance in solving the model-checking problem.

Consider an \ebflso formula $\varphi = \Ef \theta$. We say that an operator $g \in \{\vee,\wedge,E,A,X,F,G,U\}$ has a \emph{positive polarity} in $\varphi$ if all the occurrences of $g$ in $\theta$ are in a scope of an even number of negations. Dually, $g$ has a \emph{negative polarity} in $\varphi$ if all its occurrences in $\theta$ are in a scope of an odd number of negations.

In order to define \cbfls, let us first define the fragments \ecbflso of \ebflso and  \acbflso of \abflso, which constitute the inner level of \cbfls.
The logic \ecbflso is a fragment of \ebflso in which the only operators with a positive polarity are $\wedge$, $A$, $X$, and $G$, and the only operators with a negative polarity are $\vee$, $E$, $X$, and $F$. Note that $U$ is not allowed, as its semantics involves both conjunctions and disjunctions.
%gal2: Do we allow U with negative polarity? How?
%orna2 I should have warned you not to read this paragraph...
\stam{
An exception is Boolean assertions applied to atomic propositions, in which disjunction is allowed.
} %stam
%gal2: You say in the next sentence that in a positive form you don't have \vee, not even for atomic propositions.
Note that by pushing negations inside, we make all operators
\stam{
(except these applied to atomic propositions)
}
of a positive polarity; that is, we are left only with $\wedge$, $A$, $X$, and $G$.
Then, since all requirements are universal and conjunctively related, we can push conjunctions outside so that path formulas do not have internal conjunctions ---~for example, transform $AX(\xi_1 \wedge \xi_2)$ into $AX \xi_1 \wedge AX \xi_2$, and can get rid of universal path quantification that is nested inside another universal path quantification --- for example, transform $AXAX\xi_1$ into $AXX\xi$. Finally, since we use the strong semantics to $X$, we can replace formulas that have $X$ nested inside $G$ by $\false$.
The logic \acbflso is the dual fragment of \abflso. In other words, $\Af \theta$ is in \acbflso iff $\Ef \neg \theta$ is in \ecbflso.

Intuitively, the logic \cbfls is obtained by going up a hierarchy in which formulas of lower levels serve as atomic propositions in higher levels. We now define the syntax of \cbfls formally. For simplicity, we define it in a normal form, obtained by applying the rules described above.
A \cbflsz formula is a Boolean assertion over $AP$.
For $i \geq 0$, a \cbflsipo formula is a Boolean assertion over \cbflsi formulas and formulas of the form
$\Ef (A\psi_1 \wedge \cdots \wedge A\psi_n)$, where $\psi_j$ is of the form $X^{k_j}\xi_j$ or $X^{k_j}G\xi_j$, where $k_j \geq 0$ and $\xi_j$ is a \cbflsi formula or a flow proposition (that is, $> \gamma$, $< \gamma$, $\geq \gamma$, or $\leq \gamma$, for an integer $\gamma \in \N$).
%conjunction $\theta_1 \wedge \Ef \theta_2$, where $\theta_1$ is a Boolean assertion over \cbflsi formulas, and $\theta_2$ is a conjunction $A\psi_1 \wedge \cdots \wedge A\psi_n$, where $\psi_j$ is of the form $X^{k_j}\xi_j$ or $X^{k_j}G\xi_j$, where $k_j \geq 0$ and $\xi_j$ is a \cbflsi formula or a flow proposition (that is, $> \gamma$ or $\geq \gamma$, for an integer $\gamma \in \N$).
Then, a  \cbfls formula is a \cbflsi formula for some $i \geq 0$.
Note that both \ecbflso and \acbflso are contained in \cbflso.

%of stam

\begin{exa}
Recall the job-assignment problem from Example~\ref{tw example}. The \cbflso formula
%space
$\Af(<10 \rightarrow EX \leq 0) \wedge \Ef(15 \wedge AX \geq 1)$
states that if less than $10$ jobs are processed, then at least one worker is unemployed, but it is possible to process $15$ jobs and let every worker process at least one job.
\end{exa}

\vspace{3mm}
\noindent
{\bf \bfl.}
The logic \bfl is the fragment of \bfls in which every modal operator ($X,U$) is immediately preceded by a path quantifier.
That is, it is the flow counterpart of CTL\@.
As we are going to show, while in temporal logic, the transition from \ctls to CTL significantly reduces model-checking complexity, this is not the case for \bfl.

\begin{rem}%
\label{remark ep}
As demonstrated in the examples, the extensions of \bfls simplifies the specification of natural properties. We do not, however, study the expressive power of the extensions and fragments. For some, results about expressive power follow from known results about \ctls. For example, as proven in~\cite{KPV12}, adding past to \ctls with a branching-past semantics strictly increases its expressive power. The same arguments can be used in order to show that  \bfls with past operators is strictly more expressive than \bfls.\footnote{We note that the result from~\cite{KPV12} does not immediately imply the analogous result for \bfls, as it might be the case that the flow operators somehow cover up for the extra power of the branching-past temporal operator. In particular, the proof in~\cite{KPV12} is based on the fact that only \ctls with past operators is sensitive to unwinding and, as we prove in Section~\ref{unwinding}, \bfls is sensitive to unwinding. Still, since the specific formula used in~\cite{KPV12} in order to prove the sensitivity of \ctls with past to unwinding does not refer to flow, it is easy to see that it has no equivalent \bfls formula.} Likewise, it is easy to show that \bfls is strictly more expressive than \bfl.
On the other hand, it is not hard to see that quantification over the maximum flow can be expressed by first-order quantification on flow values. Indeed, $\Ef^\mathit{max} \psi$ is equivalent to
$\exists x \Ef ((=x) \wedge \psi \wedge \Af(\leq x))$. Likewise,
it may well be that every  \bfls formula with non-integral flow quantifiers has an equivalent one without such quantifiers, or that %quantification over the maximum flow
positive path quantification
can be expressed by first-order quantification on flow values.
We leave the study of the expressive power and relative succinctness of the various logics to future research.
\end{rem}

\section{Model Checking}%
\label{sec mc}
In this section we study the model-checking problem for \bfls. The problem is to decide, given a flow network $N$ and a \bfls formula $\varphi$, whether $N\models \varphi$.
\begin{thm}%
\label{pspace mc}
\bfls model checking is PSPACE-complete.
\end{thm}

\begin{proof}
Consider a network $N$ and a \bfls formula $\varphi$.  The idea behind our model-checking procedure is similar to the one that recursively employs LTL model checking in the process of \ctls model checking~\cite{EL87}. Here, however, the setting is more complicated. Indeed, the path formulas in \bfls are not ``purely linear'', as the flow quantification in them refers to flow in the (branching) network. In addition, while the search for witness paths is restricted to paths in the network, which can be guessed on-the-fly in the case of LTL, here we also search for witness flow functions, which have to be guessed in a global manner.

Let $\{\varphi_1,\ldots,\varphi_k\}$ be the set of flow state formulas in $\varphi$. Assume that $\varphi_1,\ldots,\varphi_k$ are ordered so that for all $1 \leq i \leq k$, all the subformulas of $\varphi_i$ have indices in $\{1,\ldots,i\}$. Our model-checking procedure labels $N$ by new atomic propositions $q_1,\ldots,q_k$ so that for all vertices $v$ and $1 \leq i \leq k$, we have that $v \models q_i$ iff $v \models \varphi_i$ (note that since $\varphi_i$ is closed, satisfaction is independent of a flow function).

Starting with $i=1$, we model check $\varphi_i$, label $N$ with $q_i$, and replace the subformula $\varphi_i$ in $\varphi$ by $q_i$. Accordingly, when we handle $\varphi_i$, it is an \ebflso or a \abflso formula. That is, it is of the form  $\Ef\xi$ or $\Af \xi$, for a Flow-\ctls formula $\xi$. Assume that $\varphi_i = \Ef\xi$. We guess a flow function $f:E \rightarrow \N$, and perform \ctls model-checking on $\xi$, evaluating the flow propositions in $\xi$ according to $f$. Since guessing $f$ requires polynomial space, and \ctls model checking is in PSPACE, so is handling of $\varphi_i$ and of all the subformulas.

Hardness in PSPACE follows from the hardness of \ctls model checking~\cite{SC85}.
\end{proof}

%orna3 towards the theorem:
Since \bfls contains \ctls, the lower bound in Theorem~\ref{pspace mc} is immediate. One may wonder whether reasoning about flow networks without atomic propositions, namely when we specify properties of flow only, is simpler. Theorem~\ref{pspace no ap} below shows that this is not the case. Essentially, the proof follows from our ability to encode assignments to atomic propositions by values of flow.

\begin{thm}%
\label{pspace no ap}
\bfls model checking is PSPACE-complete already for \alflo formulas without atomic propositions.
\end{thm}
\begin{proof}
We describe a reduction from LTL satisfiability, which is PSPACE-hard~\cite{SC85}. As discussed in Remark~\ref{fin paths}, PSPACE-hardness applies already to LTL over finite computations.

Consider an LTL formula $\psi$. For simplicity, we assume that $\psi$ is over a single atomic proposition $p$. Indeed, LTL satisfiability is PSPACE-hard already for this fragment. We construct a flow network $N$ and an \alflo formula $\varphi$ such that $N$ does not satisfy $\varphi$  iff some computation in ${\{p,\neg p\}}^*$ satisfies $\psi$. The network $N$ appears in Figure~\ref{figure PSPACE no AP}. Note that when $s$ has a flow of $1$, then exactly one of $v_p$ and $v_{\neg p}$ has a flow of $1$.
Let $\psi'$ be the Flow-LTL formula obtained from $\psi$ be replacing all occurrences of $p$ by $\geq 1$.
We define $\varphi= \Af A(1 \rightarrow X \neg \psi')$. Thus, $\neg \varphi= \Ef E(1 \wedge X\psi')$.

\begin{figure}[ht]
\begin{center}
%% \vspace{-5mm}
\includegraphics[scale=0.3]{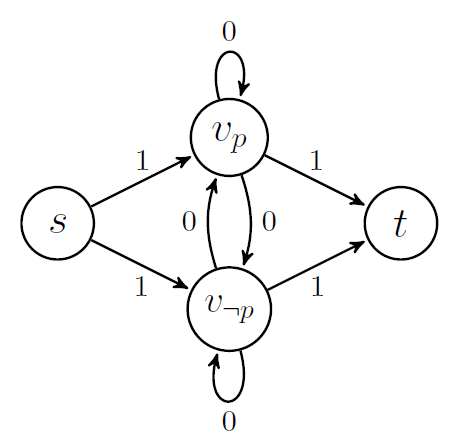}
\caption{\normalsize The flow network $N$.}%
\label{figure PSPACE no AP}
\end{center}
\end{figure}

We prove that $N$ satisfies $\Ef E(1 \wedge X\psi')$ iff some computation in ${\{p,\neg p\}}^*$ satisfies $\psi$.
Assume first that some computation $\pi$ in ${\{p,\neg p\}}^*$ satisfies $\psi$. Consider a flow $f$ in $N$ such that $f(v_p)=1$ and $f(v_{\neg p})=0$, and consider the path in $N$ that encodes $\pi$ (starting from the second vertex in the path). For these flow and path, the formula $X\psi'$ holds, and therefore $N \models \Ef E(1 \wedge X\psi')$.
For the other direction, assume that $N \models \Ef E(1 \wedge X\psi')$. Let $f$ and $\pi$ be the flow function and the path in $N$ that witness the satisfaction of $X\psi'$. The structure of $N$ guarantees that $f(v_p)+f(v_{\neg p})=1$. Hence, $\pi$ with $f$ encodes a computation in ${\{p,\neg p\}}^*$ that satisfies $\psi$, and we are done. Note that we can avoid using edges with capacities $0$ by changing $\neg \varphi$ to $\Ef E(1 \wedge X0 \wedge XX\psi')$.
\end{proof}

\stam{
\begin{figure}[t]
\begin{minipage}[b]{0.4\linewidth}
\begin{center}
%% \vspace{-5mm}
\includegraphics[scale=0.3]{PSPACE_no_AP.png}
\caption{\normalsize The flow network $N$.}%
\label{figure PSPACE no AP}
\end{center}
\end{minipage}
\quad
\begin{minipage}[b]{0.55\linewidth}
\begin{center}
%%% \vspace{-5mm}
\includegraphics[scale=0.24]{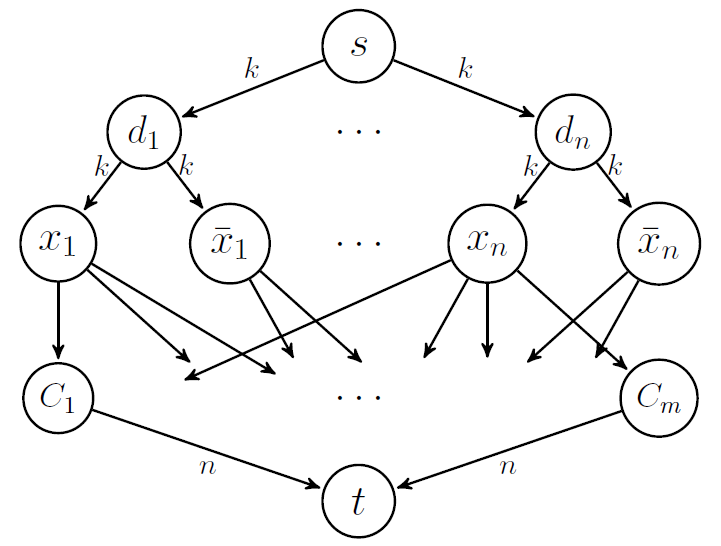}
\caption{\normalsize The flow network $N$. The capacity of each of the edges entering $C_1,\ldots,C_m$ is $1$.}%
\label{figure NP reduction}
\end{center}
\end{minipage}
\end{figure}
}
%space
\stam{
\begin{figure}[ht]
\begin{center}
%% \vspace{-5mm}
\includegraphics[scale=0.3]{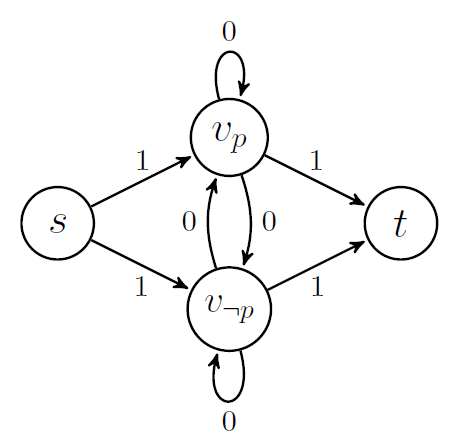}
\caption{\normalsize The flow network $N$.}%
\label{figure PSPACE no AP}
\end{center}
\end{figure}
}
We also note that when the given formula is in \bfl, we cannot avoid the need to guess a flow function, yet once the flow function is guessed, we can verify it in polynomial time. Accordingly, the model-checking problem for \bfl is in ${\rm P}^{\rm NP}$. We discuss this point further below.

In practice, a network is typically much bigger than its specification, and its size is the computational bottleneck. In temporal-logic model checking, researchers have analyzed the \emph{system complexity} of model-checking algorithms, namely the complexity in terms of the system, assuming the specification is of a fixed length. There, the system complexity of LTL and \ctls is NLOGSPACE-complete~\cite{LP85,KVW00}. We prove that, unfortunately, this is not the case of \bfls.  That is, we prove that while the \emph{network complexity} of the model-checking problem, namely the complexity in terms of the network, does not reach PSPACE, it seems to require polynomially many calls to an NP oracle. Essentially, each evaluation of a flow quantifier requires such a call.
%gal2: more than one call if the flow quantifier is inside a path formula
%orna2: added ``evaluation of''
Formally, we have the following.

\begin{thm}%
\label{Delta2 upper bound}
The network complexity of \bfls is in $\Delta_2^P$ (i.e., in ${\rm P}^{\rm NP}$).
\end{thm}

\begin{proof}
Fixing the length of the formula in the algorithm described in the proof of Theorem~\ref{pspace mc}, we get that $k$ is fixed, and so is the length of each subformula $\varphi_i$. Thus, evaluation of $\varphi_i$ involves a guess of a flow and then model checking of a fixed size Flow-\ctls formula, which can be done in time polynomial in the size of the network. Hence, the algorithm from Theorem~\ref{pspace mc} combined with an NP oracle gives the required network complexity.
\end{proof}

\stam{
%gal1: the result about $BH_{k+1} \cap co-BH_{k+1}$ upper bound is not correct. It was correct before we changed the definition of BFL* a few weeks ago. Since we need to check the subformulas for every vertex, a constant number of queries to the NP oracle is no enough. We should remove it. Maybe we can say something about a BH lower bound but I need to think about it, and it is not that important. So we need to remove this paragraph.
%orna1 I want to leave something, and we should be able to make it correct by refining the ``fragment of \bfls in which there are only $k$ flow quantifiers'' to something that talks about nesting (or whatever, I am not sure yet)
The complexity class BH is based on a Boolean hierarchy over NP\@. Essentially, it is the smallest class that contains NP and is closed under union, intersection, and complement. The levels of the hierarchy start with BH$_1=$ NP, and each level adds internal intersections as well as intersection with a co-NP (even levels) or an NP (odd levels) language~\cite{Wec85}.
BH is contained in $\Delta_2^P$. It is not hard to prove that the network complexity of the fragment of \bflso that contains at most $k$ flow quantifiers is in BH$_{k+1} \cap$ co-BH$_{k+1}$. Indeed, the latter contain  problems that are decidable in polynomial time with $k$ parallel queries to an NP oracle~\cite{Bei91}. A BH$_k$ lower bound can also be shown.
}
While finding the exact network complexity of model checking \bfls and its fragments is interesting from a complexity-theoretical point of view, it does not contribute much to our story. Here, we prove NP and co-NP hardness holds already for very restricted fragments. As good news, in Section~\ref{poly frag} we point that for the conjunctive fragment, model checking can be performed in polynomial time.\footnote{%gal1: the result about $BH_{k+1} \cap co-BH_{k+1}$ upper bound is not correct. It was correct before we changed the definition of BFL* a few weeks ago. Since we need to check the subformulas for every vertex, a constant number of queries to the NP oracle is no enough. We should remove it. Maybe we can say something about a BH lower bound but I need to think about it, and it is not that important. So we need to remove this paragraph.
%orna1 I want to leave something, and we should be able to make it correct by refining the ``fragment of \bfls in which there are only $k$ flow quantifiers'' to something that talks about nesting (or whatever, I am not sure yet)
A possible tightening of our analysis is via the complexity class BH, which is based on a Boolean hierarchy over NP\@. Essentially, it is the smallest class that contains NP and is closed under union, intersection, and complement. The levels of the hierarchy start with BH$_1=$ NP, and each level adds internal intersections as well as intersection with a co-NP (even levels) or an NP (odd levels) language~\cite{Wec85}.
BH is contained in $\Delta_2^P$. It is not hard to prove that the network complexity of the fragment of \bflso that contains at most $k$ flow quantifiers is in BH$_{k+1} \cap$ co-BH$_{k+1}$. Indeed, the latter contain  problems that are decidable in polynomial time with $k$ parallel queries to an NP oracle~\cite{Bei91}. A BH$_k$ lower bound can also be shown.}

We first show a simple lemma that is needed for the network-complexity hardness results.

\begin{lem}%
\label{lem:CNF}
A CNF formula $\psi$ can be translated to a CNF formula $\psi'$ in polynomial time, such that $\psi$ is satisfiable iff $\psi'$ is satisfiable, and all literals in $\psi'$ appear exactly the same number of times in $\psi'$.
%Given a CNF formula $\psi$, it can be translated to another CNF formula $\psi'$ in polynomial time such that all literals in $\psi'$ appear exactly the same number of times in $\psi'$ and $\psi$ is satisfiable if and only if $\psi'$ is satisfiable.
\end{lem}
\begin{proof}
We consider CNF formulas in which a literal is not allowed to appear more than once in each clause.
Let $|y|_{\psi}$ denote the number of occurrences of the literal $y$ in $\psi$.
Let $x_1$ and $x_2$ be two variables in $\psi$.
We add some `{\sf true}' clauses to $\psi$ to obtain $\psi'$.
First we add clauses of the form ($x_1 \vee x_2 \vee \neg x_2$) for $|\neg x_1|_{\psi}-|x_1|_{\psi}$ times if $|\neg x_1|_{\psi}>|x_1|_{\psi}$, else we add clauses of the form ($\neg x_1 \vee x_2 \vee \neg x_2$) for $|x_1|_{\psi} - |\neg x_1|_{\psi}$ times.
Thus we obtain  a formula $\psi_1$ with the same number of occurrences for the literals $x_1$ and $\neg x_1$.

Now for every variable $x_i \neq x_1$, we add clauses ($x_i \vee x_1 \vee \neg x_1$) if the number of occurrences of $\neg x_i$ in $\psi_1$ is more than that of $x_i$, otherwise we add clauses ($\neg x_i \vee x_1 \vee \neg x_1$), possibly multiple times, until the number of occurrences of $x_i$ and the number of occurrences of $\neg x_i$ become the same.
Now for each variable $x$, both the literals $x$ and $\neg x$ appear the same number of times.

As a final step, let $x_j$ be a variable such that the number of occurrences of the literals $x_j$ and $\neg x_j$ is the maximum over all the variables.
For every variable $x_i \neq x_j$, we add clauses of the form ($x_i \vee \neg x_i$) until the number of literals for $x_i$ and $x_j$ are the same.
We call the resultant formula $\psi'$.

Note that all the above steps can be done in polynomial time and the size of $\psi'$ is polynomial in the size of $\psi$.
\end{proof}

\begin{thm}%
\label{ebflso NP complete}
The network complexity of \ebflso and \abflso is NP-complete and co-NP-complete, respectively. Hardness applies already to \elflo and \alflo formulas
without atomic propositions, and to BFL\@.
\end{thm}
\begin{proof}
For the upper bound, it is easy to see that one step in the algorithm described in the proof of Theorem~\ref{pspace mc} (that is, evaluating $\varphi_i$ once all its flow state subformulas have been evaluated), when applied to $\varphi_i$ of a fixed length is in NP for $\varphi_i$ of the form $\Ef \xi$ and in co-NP for $\varphi_i$ of the form $\Af \xi$.

For the lower bound, we prove NP-hardness for \elflo. Co-NP-hardness for \alflo follows by dualization. We describe a reduction from CNF-SAT\@. Let $\theta=C_1 \wedge \ldots \wedge C_m$ be a CNF formula over the variables $x_1 \ldots x_n$. We assume that every literal in $x_1,\ldots,x_n,\bar{x}_1,\ldots,\bar{x}_n$ appears exactly in $k$ clauses in $\theta$. Indeed, by Lemma~\ref{lem:CNF}, every CNF formula can be converted to such a formula in polynomial time and with a polynomial blowup.

\noindent
\begin{minipage}{0.55\textwidth}
\vspace{1mm}
We construct a flow network $N$ and an \elflo formula $\Ef A\psi$ such that $\theta$ is satisfiable iff $N \models \Ef A \psi$. The network $N$ is constructed as demonstrated on the right. %Figure~\ref{figure NP reduction}.
Let $Z=\{x_1,\ldots,x_n,\bar{x}_1,\ldots,\bar{x}_n\}$. For a literal $z \in Z$ and a clause $C_i$, the network $N$ contains an edge $\zug{z,C_i}$ iff the clause $C_i$ contains the literal $z$. Thus, each vertex in $Z$ has exactly $k$ outgoing edges. The capacity of each of these edges is $1$. The flow-LTL formula $\psi=kn \wedge XX(k \vee 0) \wedge XXX(\geq 1)$.
%In Appendix~\ref{app ebflso NP complete},
%In the full version,
We now prove that $\theta$ is satisfiable iff $N \models \Ef A \psi$.\end{minipage} \ \hspace{.2in}
\begin{minipage}{0.40\textwidth}
%% \vspace{-4mm}
\includegraphics[scale=0.3]{NP_reduction.png}
\end{minipage}

\stam{
For the lower bound, we prove NP-hardness for \elflo. Co-NP-hardness for \alflo follows by dualization. We describe a reduction from CNF-SAT\@. Let $\theta=C_1 \wedge \ldots \wedge C_m$ be a CNF formula over the variables $x_1 \ldots x_n$. We assume that every literal in $x_1,\ldots,x_n,\bar{x}_1,\ldots,\bar{x}_n$ appears exactly in $k$ clauses in $\theta$. Indeed, every CNF formula can be converted to such a formula in polynomial time and with a polynomial blowup. We construct a flow network $N$ and an \elflo formula $\Ef A\psi$ such that $\theta$ is satisfiable iff $N \models \Ef A \psi$. The network $N$ is constructed as demonstrated in Figure~\ref{figure NP reduction}. Let $Z=\{x_1,\ldots,x_n,\bar{x}_1,\ldots,\bar{x}_n\}$. For a literal $z \in Z$ and a clause $C_i$, the network $N$ contains an edge $\zug{z,C_i}$ iff the clause $C_i$ contains the literal $z$. Thus, each vertex in $Z$ has exactly $k$ outgoing edges. The capacity of each of these edges is $1$. The flow-LTL formula $\psi=kn \wedge XX(k \vee 0) \wedge XXX(\geq 1)$.

In Appendix~\ref{app ebflso NP complete}, we prove that $\theta$ is satisfiable iff $N \models \Ef A \psi$.

\begin{figure}[ht]
\begin{center}
%%% \vspace{-5mm}
\includegraphics[scale=0.3]{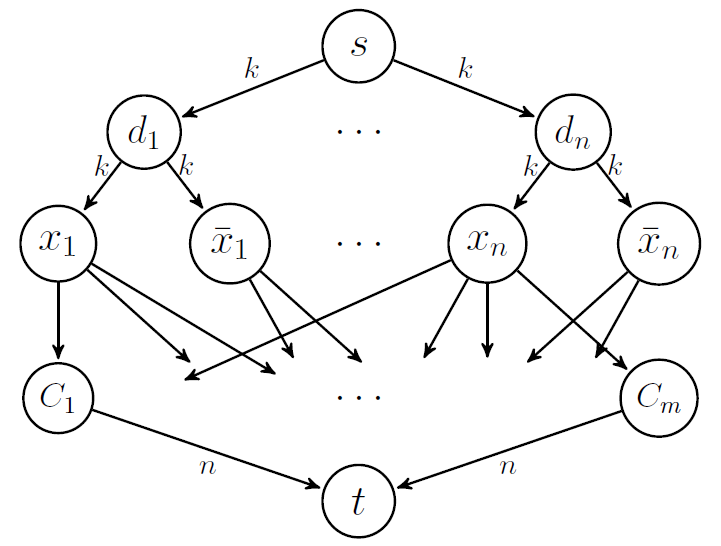}
\caption{\normalsize The flow network $N$. The capacity of each of the edges entering $C_1,\ldots,C_m$ is $1$.}%
\label{figure NP reduction}
\end{center}
\end{figure}
} %stam

\vspace{1mm}
Assume first that $\theta$ is satisfiable. Consider the flow function $f$ where for every $1 \leq i \leq n$, if $x_i$ holds in the satisfying assignment, then $f(x_i)=k$ and $f(\bar{x}_i)=0$, and otherwise $f(x_i)=0$ and $f(\bar{x}_i)=k$. For such $f$, all the paths in $N$ satisfy $\psi$.

Assume now that there is a flow function $f$ with which $N$ satisfies $A\psi$. For every $1 \leq i \leq n$, either we have $f(x_i)=k$ and $f(\bar{x}_i)=0$ or we have $f(x_i)=0$ and $f(\bar{x}_i)=k$. Consider the assignment $\tau$ to the variables $x_1,\ldots,x_n$ induced by $f$, where a literal $z \in Z$ holds in $\tau$ iff the flow in the corresponding vertex is positive. Since according to $\psi$ for every $i$ the flow in the vertex $C_i$ is positive, then every clause in $\theta$ contains at least one literal whose corresponding vertex has a positive flow. Hence, the assignment $\tau$ satisfies $\theta$.
%}
Finally, note that $\psi$ does not contain atomic propositions. Also, the same proof holds with the BFL formula $\psi=kn \wedge AXEXk  \wedge AXAXAX(\geq 1)$.
\end{proof}

\stam{
\begin{thm}
Let \bflsok be the fragment of \bflso in which at most $k$ flow quantifiers are allowed. The network complexity of \bflsok is in $BH_{k+1} \cap co-BH_{k+1}$ and is $BH_k$-hard and co-$BH_k$-hard. Hence, \bflso is hard for every class in the Boolean hierarchy.
\end{thm}
\begin{proof}
We start with the upper bound. We say that a problem is decidable in polynomial time with $j$ parallel queries to an NP oracle if it can be decided by a polynomial time oracle Turing machine that prepares a list of the $j$ queries it is going to make to the NP oracle before actually making any of them. It is known that the class of problems that are decidable in polynomial time with $j$ parallel queries to an NP oracle is contained in $BH_{j+1} \cap co-BH_{j+1}$~\cite{Bei91}. Hence, it is enough to show that model-checking for a fixed side \bflsok formula can be done in polynomial time with $k$ parallel queries to an NP oracle. Let $N$ be a flow network and let $\psi$ be a \bflsok formula. We assume that $\psi$ contains only existential flow quantifiers. Indeed, every \bfls formula can be translated to such formula. Note that $\psi$ is a propositional combination of $k$ \ebflso sub-formulas. Since according to Theorem~\ref{ebflso NP complete} solving the problem for an \ebflso formula is in NP, we use $k$ parallel queries to the NP oracle in order to solve the problem for each of the \ebflso sub-formulas. Now, checking whether $N \models \psi$ can be done easily in polynomial time.

We now show the lower bound.
\end{proof}
}

\subsection{Flow synthesis}
In the flow-synthesis problem, we are given a network $N$ and an \ebflso formula $\Ef \varphi$, and we have to return a flow function $f$ with which $\varphi$ is satisfied in $N$, or declare that no such function exists. The corresponding decision problem, namely, deciding whether there is a flow function with which $\varphi$ is satisfied in $N$, is clearly at least as hard as \ctls model checking. Also, by guessing $f$, its complexity does not go beyond \ctls model-checking complexity. The network complexity of the problem coincides with that of \ebflso model checking. Thus, we have the following.

\begin{thm}
The flow-synthesis problem for \ebflso is PSPACE-complete, and its network complexity is NP-complete.
\end{thm}

\section{Model Checking Extensions of \texorpdfstring{\bfls}{BFL*}}%
\label{sec mce}

In Section~\ref{exts}, we defined several extensions of \bfls. In this section we study the model-checking complexity for each of the extensions, and show that they do not require an increase in the complexity. The techniques for handling them are, however, richer: For positive path quantification, we have to refine the network and add a path-predicate that specifies positive flow, in a similar way fairness is handled in temporal logics. For maximal-flow quantification, we have to augment the model-checking algorithm by calls to a procedure that finds the maximal flow. For non-integral flow quantification, we have to reduce the model-checking problem to a solution of a linear-programming system. For past operators, we have to extend the model-checking procedure for \ctls with branching past. Finally, for first-order quantification over flow values, we have to first bound the range of relevant values, and then apply model checking to all relevant values.

\subsection{Positive path quantification}%
\label{ppqa}
Given a network $N$, it is easy to generate a network $N'$ in which we add a vertex in the middle of each edge, and in which the positivity of paths correspond to positive flow in the new intermediate vertices. Formally, assuming that we label the new intermediate vertices by an atomic proposition $\mathit{edge}$, then the \bfls path formula $\xi_\mathit{positive}=G(\mathit{edge} \rightarrow >0)$ characterizes positive paths, and replacing a state formula $A\psi$ by the formula $A(\xi_\mathit{positive} \rightarrow \psi)$ restricts the range of path quantification to positive paths.
Now, given a \bfls formula $\varphi$, it is easy to generate a \bfls formula $\varphi'$ such that $N \models \varphi$ iff $N' \models \varphi'$. Indeed, we only have to (recursively) modify path formulas so that vertices labeled $\mathit{edge}$ are ignored: $X\xi$ is replaced by $XX\xi$, and  $\xi_1 U \xi_2$ is replaced by $(\xi_1 \vee \mathit{edge})U(\xi_2 \wedge \neg \mathit{edge})$.
Hence, the complexity of model-checking is similar to \bfls.

\subsection{Maximal flow quantification}%
\label{mfqa}
\stam{
Consider a vertex $v \in V$ in which a flow formula $\Ef^\mathit{max} \varphi$ is evaluated. We can find in polynomial time the value $\gamma_\mathit{max}$ of the maximal flow from $v$ to $T$. Then, whenever the algorithm guesses a flow function, we can restrict attention to functions whose value is $\gamma_\mathit{max}$.
Also, yak yak.
} %stam
The maximal flow $\gamma_\mathit{max}$ in a flow network can be found in polynomial time. Our model-checking algorithm for \bfls described in the proof of Theorem~\ref{pspace mc} handles each flow state subformula $\Ef \varphi$ by guessing a flow function $f:E \rightarrow \N$ with which the Flow-\ctls formula $\varphi$ holds. For an $\Ef^\mathit{max}$ quantifier, we can guess only flow functions for which the flow leaving the source vertex is $\gamma_\mathit{max}$. In addition, after calculating the maximal flow, we can substitute $\gamma_\mathit{max}$, in formulas that refer to it, by its value. Hence, the complexity of model-checking is similar to that of \bfls.

\subsection{Non-integral flow quantification}%
\label{with ni flow algo}
Recall that our \bfls model-checking algorithm handles each flow state subformula $\Ef \varphi$ by guessing a flow function $f:E \rightarrow \N$ with which the Flow-\ctls formula holds. Moving to non-integral flow functions, the guessed function $f$ should be $f:E \rightarrow \R$, where we cannot bound the size or range of guesses.

Accordingly, in the non-integral case, we guess, for every vertex $v \in V$, an assignment to the flow propositions that appear in $\varphi$. Then, we perform two checks. First, that $\varphi$ is satisfied with the guessed assignment --- this is done by \ctls model checking, as in the case of integral flows. Second, that there is a non-integral flow function that satisfies the flow constraints that appear in the vertices. This can be done in polynomial time by solving a system of linear inequalities~\cite{Sch03}. %(see Lemma~\ref{cfp} for the details in the case of vertex-constrained integral flow functions).
Thus, as in the integral case, handling each flow state formula $\Ef \varphi$ can be done in PSPACE, and so is the complexity of the entire algorithm.

\subsection{Past operators}%
\label{mc with past}
Recall that our algorithm reduces \bfls model checking to a sequence of calls to a \ctls model-checking procedure. Starting with a \bfls formula with past operators, the required calls are to a model-checking procedure for \ctls with past.
By~\cite{KPV12}, model checking \ctls with branching past is PSPACE-complete, and thus so is the complexity of our algorithm.

\subsection{First-Order quantification on flow values}%
\label{foqa}

For a flow network
\[
    N=\zug{AP, V,E,c,\rho,s,T},
\]
let $C_N=1+\Sigma_{e \in E}c(e)$. Thus, for every flow function $f$ for $N$ and for every vertex $v \in V$, we have $f(v)<C_N$. We claim that when we reason about \bfls formulas with quantified flow values, we can restrict attention to values in $\{0,1,\ldots,C_N\}$.

\begin{lem}%
\label{finite forall}
Let $N$ be a flow network and let $\theta=\forall x_1 \varphi$ be a BFL$^\star(\{+,*\})$ formula over the variables $X=\{x_1,\ldots,x_n\}$, and without free variables. Then, $N \models \theta$ iff $N \models \varphi[x_1 \leftarrow \gamma]$, for every $0 \leq \gamma \leq C_N$.
\end{lem}
\begin{proof}
We show that for every $\gamma>C_N$ we have $N \models \varphi[x_1 \leftarrow \gamma]$ iff $N \models \varphi[x_1 \leftarrow C_N]$.
Every expression $g(X)$ in $\varphi$ is obtained from the variables in $X$ by applying the operators $+,*$, and possibly by using constants in $\N$. Hence, every such $g$ is monotonically increasing for every variable $x_i$. Let $\gamma_2,\ldots,\gamma_n$ be an assignment for the variables $x_2,\ldots,x_n$, respectively. Let $g_1(x_1)=g(x_1,\gamma_2,\ldots,\gamma_n)$. Let $\gamma_1$ be an integer bigger than $C_N$. Since $g_1$ is monotonically increasing, we have $g_1(C_N) \leq g_1(\gamma_1)$. If $g_1(C_N) < g_1(\gamma_1)$, then $g_1$ is not a constant function, and since it contains only the operators $+,*$ and constants in $\N$, we have $g_1(C_N) \geq C_N$. Therefore, if $g_1(C_N) \neq g_1(\gamma_1)$, then $g_1(\gamma_1) > g_1(C_N) \geq C_N$. Recall that for every flow function $f$ in $N$ and for every vertex $v \in V$, we have $f(v)<C_N$. Therefore, every state subformula $>g_1(x_1)$ and every state subformula $\geq g_1(x_1)$ holds for $\gamma_1$ iff it holds for $C_N$. Therefore, $N \models \varphi[x_1 \leftarrow \gamma]$ iff $N \models \varphi[x_1 \leftarrow C_N]$. Hence, we have $N \models \forall x_1 \varphi$ iff $N \models \varphi[x_1 \leftarrow \gamma]$ for every $0 \leq \gamma \leq C_N$.
\end{proof}

By Lemma~\ref{finite forall}, the model-checking problem for BFL$^\star(\{+,*\})$ is PSPACE-complete.

%of stam

\section{A Polynomial Fragment}%
\label{poly frag}

The upper bounds to the complexity of the model-checking problem that are given in Sections~\ref{sec mc} and~\ref{sec mce} refer to worst-case complexity. As has been the case of \ctls, in practice the complexity is often lower. In particular, there are fragments of \ctls, and hence also \bfls and its extensions for which we can show that the described model-checking algorithms perform better, say by showing that the particular syntax guarantees a bound on the size of the automata associated with path formulas. In this section we show that the model-checking problem for \cbfls (see Section~\ref{frags}) can be solved in polynomial time. The result is not based on a tighter complexity analysis for the case the specification is in \cbfls but rather on a different model-checking algorithm for this logic.

\stam{
We first define a \emph{bold form} for \cbfls formulas.
Recall that a \cbflsz formula is a Boolean assertion over $AP$. All \cbflsz formula are bold.
For $i \geq 0$, a bold \cbflsipo formula is a conjunction $\theta_1 \wedge \Ef \theta_2$, where $\theta_1$ is a Boolean assertion over bold \cbflsi formulas, and $\theta_2$ is a conjunction $A\psi_1 \wedge \cdots \wedge A\psi_n$, where $\psi_j$ is of the form $X^{k_j}\xi_j$, $X^{k_j}G\xi_j$, where $k_j \geq 0$ and $\xi_j$ is a bold \cbflsi formula or a flow proposition (that is, $> \gamma$ or $\geq \gamma$, for an integer $\gamma \in \N$).
Then, a bold \cbfls formula is a bold \cbflsi formula for some $i \geq 0$.

\begin{lem}%
\label{bold nf}
Every \cbfls formula can be translated in linear time to an equivalent bold \cbfls formula.
%gal2: note that according to your definition, multiple appearances of the quantifier \Ef are allowed. For example, in \cbfls_1 you can have a formula \Ef\varphi_1 \wedge \Ef\varphi_2. Nesting of \Ef is also allowed.
\end{lem}

\begin{proof}
Yak yak
\end{proof}
}
Our model-checking algorithm reduces the evaluation of a \cbfls formula into a sequence of solutions to the  \emph{ vertex-constrained flow problem}. In this problem, we are given a flow network $N=\zug{AP, V,E,c,\rho,s,T}$ in which each vertex $v \in V$ is attributed by a range $[\gamma_l,\gamma_u] \in \N \times (\N \cup \{\infty\})$.
%gal2: \gamma_u may be infinite
The problem is to decide whether there is a flow function $f:E\rightarrow \R$ such that for all vertices $v \in V$, we have $\gamma_l \leq f(v) \leq \gamma_u$.

\begin{lem}%
\label{cfp}
The vertex-constrained flow problem can be solved in polynomial time. If there is a solution that is a non-integral flow function, then there is also a solution that is an integral flow function, and the algorithm returns such a solution.
%gal3: you defined the problem only for integer flows, so you don't need to mention that the flow the algorithm finds is integer. If you want to mention it you should say that if there is a non-integral flow then there is an integral flow which is found by the algorithm.
\end{lem}
\begin{proof}
Given $N$, we construct a new flow network $N'$ in which the attributions are on the edges. Thus, $N'$ is a flow network with lower bounds on the flow. We construct $N'$ by splitting every vertex $v \in V$ with attribution $[\gamma_l,\gamma_u]$ into two vertices in $N'$: $v_{in}$ and $v_{out}$. The vertices are connected by an edge $\zug{v_{in},v_{out}}$ with attribution $[\gamma_l,\gamma_u]$. For every edge $\zug{u,v} \in E$, we add to $N'$ an edge $\zug{u_{out},v_{in}}$ with attribution $[0,c(\zug{u,v})]$. It is well known that deciding whether there exists a feasible flow and finding such a flow in a flow network with lower bounds can be done in polynomial time (see, for example, Chapter~6.7 in~\cite{AMO93}). Also, the algorithm for finding such a feasible flow reduces the problem to a maximum-flow problem. Accordingly, if there is a feasible flow then the algorithm returns an integral one.
\end{proof}

\begin{thm}
\cbfls model checking can be solved in polynomial time.
\end{thm}

\begin{proof}
Let  $N=\zug{AP, V,E,c,\rho,s,T}$, and consider a \cbfls formula $\varphi$. If $\varphi$ is in \cbflsz, we can clearly label in linear time all the vertices in $N$ by a fresh atomic proposition $p_\varphi$ that maintains the satisfaction of $\varphi$. That is, in all vertices $v \in V$, we have that $p_\varphi \in \rho(v)$ iff $v \models \varphi$.
\stam{
Otherwise, by Lemma~\ref{bold nf}, we can assume that $\varphi$ is a bold \cbflsipo formula for some $i \geq 0$.
}
Otherwise, $\varphi$ is a \cbflsipo formula for some $i \geq 0$.
We show how, assuming that the vertices of $N$ are labeled by atomic propositions that maintain satisfaction of the subformulas of $\varphi$ that are \cbflsi formulas, we can label them, in polynomial time,  by a fresh atomic proposition that maintains the satisfaction of $\varphi$.

Recall that $\varphi$ is a Boolean assertion over \cbflsi formulas and flow formulas of the form
$\Ef (A\psi_1 \wedge \cdots \wedge A\psi_n)$, where each $\psi_j$ is of the form $X^{k_j}\xi_j$ or $X^{k_j}G\xi_j$, where $k_j \geq 0$ and $\xi_j$ is a \cbflsi formula or a flow proposition (that is, $> \gamma$, $< \gamma$, $\geq \gamma$, or $\leq \gamma$, for an integer $\gamma \in \N$).

\stam{
Since \cbflsi formulas have already been evaluated in all vertices, we can easily evaluate the subformulas of $\varphi$ that are \cbflsi formulas. We describe how to evaluate flow formulas. Let $\theta=\Ef (A\psi_1 \wedge \cdots \wedge A\psi_n)$ be such a formula.
}
Since \cbflsi subformulas have already been evaluated, we describe how to evaluate subformulas of the form $\theta=\Ef (A\psi_1 \wedge \cdots \wedge A\psi_n)$.
Intuitively, since the formulas in $\theta$ include no disjunctions, they impose constraints on the vertices of $N$ in a deterministic manner. These constraints can be checked in polynomial time by solving a vertex-constrained flow problem. Recall that for each $1 \leq j \leq n$, the formula $\psi_j$ is of the form $X^{k_j}\xi_j$ or $X^{k_j}G\xi_j$, for some $k_j \geq 0$, and a \cbflsi formula or a flow proposition $\xi_j$. In order to evaluate $\theta$ in a vertex $v \in V$, we proceed as follows. For each $1 \leq j \leq n$, the formula $\xi_j$ imposes either a Boolean constraint (in case $\xi_j$ is a \cbflsi formula) or a flow constraint (in case $\xi_j$ is a flow proposition) on a finite subset $V^v_j$ of $V$. Indeed, if $\psi_j=X^{k_j}\xi_j$, then $V^v_j$ includes all the vertices reachable from $v$ by a path of length $k_j$,
and if $\psi_j=X^{k_j}G\xi_j$, then $V^v_j$ includes all the vertices reachable from $v$ by a path of length at least $k_j$. We attribute each vertex by the constraints imposed on it by all the conjuncts in $\theta$. If one of the Boolean constrains does not hold, then $\theta$
does not hold in $v$. Otherwise, we obtain a set of flow constraints for each vertex in $V$.
For example, if $\theta=\Ef(AXXp \wedge AXX>5 \wedge AG\leq 8)$, then in order to check whether $\theta$ holds in $s$, we assign the flow constraint $\leq 8$ to all the vertices reachable from $s$, and  assign the flow constraint $>5$ to all the successors of the successors of $s$. If one of these successors of successors does not satisfy $p$, we can skip the check for a flow and conclude that $s$ does not satisfy $\theta$. Otherwise, we search for such a flow, as described below.

The flow constrains for a vertex induce a closed, open, or half-closed range. The upper bound in the range may be infinity. For example, the constrains $>6$, $<10$, $\leq 8$ induce the half-closed range $\lopen{6,8}$. Note that it may be that the induced range is empty. For example, the constraints $\leq 6$ and $>8$ induce an empty range. Then, $\theta$
does not hold in $v$. Since we are interested in integral flows, we can convert all strict bounds to non-strict ones. For example, the range $\lopen{6,8}$ can be converted to $[7,8]$. Note that since we are interested in integral flow, a non-empty open range may not be satisfiable, and we refer to it as an empty range. For example, the range $(6,7)$ is empty. Hence, the satisfaction of $\theta$ in $v$ is reduced to an instance of the vertex-constrained flow problem. By Lemma~\ref{cfp}, deciding whether there is a flow function that satisfies the constraints can be solved in polynomial time.
\stam{
, and in case the answer is positive, there is an integral flow function.
}
\end{proof}

%stam

\begin{rem}
Note that the same algorithm can be applied when we consider non-integral flow functions, namely in \cbfls with the $\Af^{\R}$ flow quantifier. There, the induced vertex-constrained flow problem may include open boundaries. The solution need not be integral, but can be found in polynomial time by solving a system of inequalities~\cite{Sch03}.
\end{rem}

\stam{
\begin{thm}
Let $\psi$ be a \rbfls formula and let $N$ be a flow network. Deciding whether $N \models \psi$ can be done in polynomial time.
\end{thm}
\begin{proof}
If $\psi$ is a \rabfls formula then checking whether $N \models \psi$ can be done in polynomial time by checking whether $N \not \models \neg \psi$, since $\neg \psi$ can be written as a \rebfls formula.
If $\psi$ is a \rbfls formula then checking whether $N \models \psi$ can be done similarly to the algorithm described in the proof of Theorem~\ref{Delta2 upper bound}. In Theorem~\ref{Delta2 upper bound}, model checking of an inner \ebflso or \abflso formula is done by an NP oracle. Since $\psi$ is a \rbfls formula then the inner formulas in our case are \rebfls and \rabfls formulas and thus the model checking can be done in polynomial time.
\end{proof}
}

\section{Query Checking}%
\label{qc}

As discussed in Section~\ref{intro}, query checking is a useful methodology for system exploration.
A query is a specification with the place-holder ``$?$'', and the goal is to find replacements to $?$ with which the specification holds. In this section we extend the methodology to \bfls. We first need some definitions. % chktex 40

A \emph{propositional \bfls query} (propositional query, for short) is a \bfls formula in which a single state formula is ``$?$''. % chktex 40
For a propositional query $\psi$ and a propositional assertion $\theta$ over $AP$, we denote by $\psi[? \leftarrow \theta]$ the formula obtained from $\psi$ by replacing $?$ by $\theta$. % chktex 40
Given a network $N$ and a query $\psi$,
a \emph{solution} to $\psi$ in $N$ is a propositional assertion $\theta$ over $AP$ such that $N \models \psi[? \leftarrow \theta]$. For example, the propositional query
%gal7: I changed this example
%$\Ef E((? U (\mathit{target} \wedge (\geq 20)))$
$\Ef E((? \wedge \geq 5) U \mathit{target})$ asks which propositional assertions $\theta$ are such that the network satisfies
%$\Ef E(\theta U (\mathit{target} \wedge (\geq 20)))$
$\Ef E((\theta \wedge \geq 5) U \mathit{target})$,
%orna5
%gal5 the formula does not mean that. It means that a flow of 20 can reach a target even if we require that there is a path to the target in which for all vertices until the target \theta holds.
%orna6 This is how I interpret what we had. Anyway, you can reword.
%gal6: done
%namely propositional assertions $\theta$ such that if we restrict the flow to travel only through vertices that satisfy $\theta$, we can get a flow of $20$ units.
%namely propositional assertions $\theta$ such that a flow of $20$ can reach a target even if we require that there is a path to the target in which for all vertices until the target $\theta$ holds.
namely propositional assertions $\theta$ such that there is a flow and there is a path to the target in which all vertices satisfy $\theta$ and the flow in them is at least $5$.
Note that we do not allow the solution $\theta$ to include flow propositions.
Note also that only a single occurrence of $?$ is allowed in the query. Richer settings in \ctls query checking allow queries with several place holders, namely $?_1,?_2,\ldots,?_m$, possibly with multiple appearances to each of them~\cite{GC02}. The focus on queries with a single $?$ enables the query checker to refer to the \emph{polarity} of queries. % chktex 40
%orna5 plase find a reference
%gal5 done
Formally, we say that a propositional query $\psi$ is \textit{positive} (\emph{negative}) if the single $?$ in it is in the scope of an even (respectively, odd) number of negations. % chktex 40
The polarity of queries implies monotonicity, in the following sense (the proof proceeds by an easy induction on the structure of the query).
\begin{lem}%
\label{lemma mono}
Consider a positive (negative) \bfls query $\psi$. Let $\theta$ and $\theta'$ be propositional formulas such that $\theta$ implies $\theta'$ ($\theta'$ implies $\theta$, respectively). Then, $\psi[? \leftarrow \theta]$ implies $\psi[? \leftarrow \theta']$.
\end{lem}

\begin{cor}%
\label{lemma prop fl query}
Consider a network $N$. Let $\psi$ be a positive (negative) \bfls query, let $\theta$ be a solution for $\psi$ in $N$, and let $\theta'$ be a propositional formula such that $\theta$ implies $\theta'$ ($\theta'$ implies $\theta$, respectively). Then, $\theta'$ is a solution for $\psi$ in $N$.
\end{cor}

It is not hard to see that a propositional \bfls query may have multiple solutions.  Corollary~\ref{lemma prop fl query} enables us to partially order them. Then, given a query $\psi$, we say that a solution $\theta$ is \emph{strongest} if there is no solution $\theta'$ such that $\theta'$ is not equivalent to $\theta$ and either $\psi$ is positive, in which case $\theta'$ implies $\theta$, or $\psi$ is negative, in which case $\theta$ implies $\theta'$.
%gal7: this definition for a strongest solution is good for positive queries. If the query is negative then a strongest solution is a solution that does not imply other solutions.
A query, however, may not only have multiple solutions, but may also have multiple strongest solutions. Accordingly, as is the case with \ctls, a search for all strongest solutions has to examine all possible solutions~\cite{BG01}. Consider a network $N$ and a propositional \bfls query $\psi$. Each solution to $\psi$ in $N$ corresponds to an assignment to all subsets of atomic propositions,
%gal7: each solution corresponds to an assignment to every subset of atomic propositions
thus there are $2^{2^{|AP|}}$ assertions to check. For each we need to model check the formula $\psi[? \leftarrow \theta]$, for an assertion $\theta$. Accordingly, the length of $\psi[? \leftarrow \theta]$ is $O(|\psi|+2^{|AP|})$. Thus, finding all strongest solutions to a propositional query is a very complex task, involving $2^{2^{|AP|}}$ executions of the algorithm described in Theorem~\ref{pspace mc} for a formula of length $O(|\psi|+2^{|AP|})$.

Fortunately, in the context of \bfls, we are able to point to a class of interesting queries for which query checking is not more complex than model checking: queries in which the place holder is the value in a flow proposition.  Formally, a \emph{value \bfls query} (value query, for short) is a \bfls formula in which a single flow proposition is of the form $> ?$, $\geq ?$, $< ?$, or $\leq ?$. % chktex 40
Each value query $\psi$ is either a \emph{lower-bound query}, in case $\psi$ is positive iff the flow proposition with $?$ is $>?$ or $\geq ?$ (that is, either $\psi$ is positive and the flow proposition with $?$ is $>?$ or $\geq ?$, or $\psi$ is negative and the flow proposition with $?$ is $<?$ or $\leq ?$), or an \emph{upper-bound query}, in case $\psi$ is positive iff the flow proposition with $?$ is $<?$ or $\leq ?$. % chktex 40
For a value query $\psi$ and an integer $\gamma \in \N$, we denote by $\psi[? \leftarrow \gamma]$ the \bfls formula obtained from $\psi$ by replacing $?$ by $\gamma$. % chktex 40
Given a network $N$ and a value query $\psi$,
a \emph{solution} to $\psi$ in $N$ is an integer $\gamma \in \N$ such that $N \models \psi[? \leftarrow \gamma]$. For example, the value query
%gal7: fixed it
%$\Ef E(\mathit{target} \wedge (\geq ?)))$
$\Ef EF(\mathit{target} \wedge (\geq ?))$ asks which values of flow can reach a target vertex.

For a value query $\psi$, we say that a solution $\gamma$ to $\psi$ in $N$ is \emph{strongest} if $\psi$ is a lower-bound query and there is no solution $\gamma'$ to $\psi$ in $N$ such that $\gamma' > \gamma$, or $\psi$ is an upper-bound query and there is no solution $\gamma'$ to $\psi$ in $N$ such that $\gamma' < \gamma$.
Since a strongest solution in a value query is either a maximum or a minimum of a set in $\N$, there is at most one strongest solution.\footnote{Note that if we allow solutions in $\R$, then flow propositions with strict inequality do not have unique strongest solutions. In the case of solutions in $\N$, however, the flow propositions $< \gamma$ and $> \gamma$ are equivalent to $\leq \gamma-1$ and $\geq \gamma+1$, respectively.} Hence the following lemma.
%gal7: this lemma does not follow from the fact that there is a single strongest solution. It follows from similar arguments as in \ref{lemma mono}.

\begin{lem}%
\label{lemma value query}
Consider a network $N$. Let $\gamma$ be the strongest solution for a value \bfls query $\psi$. If $\psi$ is a lower-bound query, then the set of solutions is $\N \cap \ropen{\gamma,\infty}$. If $\psi$ is an upper-bound query then the set of solutions is $\N \cap [0,\gamma]$.
\end{lem}

\begin{rem}%
\label{no =}
Note that in our definition of a value query, we do not allow a flow proposition of the form $=?$. % chktex 40
Indeed, such a flow proposition encodes a conjunction of two flow propositions of different polarities.
If we decide to allow such value queries then the set of solutions does not satisfy a property as in Lemma~\ref{lemma value query}. For example, the set of solutions for the query $\Ef[ (=?) \wedge ((\geq 2 \, \wedge \leq 4) \vee (\geq 6 \, \wedge \leq 9)) ]$ and a network with a maximal flow
of $9$, is $\{2,3,4,6,7,8,9\}$.
\end{rem}

Thus, unlike the case of propositional queries, in the case of value queries we can talk about the strongest solution, and study the complexity of finding it. The lower bound in the following theorem corresponds to the problem of deciding whether there is a solution to a given value \bfls query.

\begin{thm}
The problem of finding the strongest solution to a value \bfls query is PSPACE-complete. Furthermore, the network complexity is in $\Delta_2^P$.
\end{thm}
\begin{proof}
Model checking of a \bfls formula $\psi$ can be reduced to deciding whether there is a solution for the value query $\psi \wedge (\Af \leq ?)$. Hence the PSPACE-hardness.

Lemma~\ref{lemma value query} implies that in order to find the strongest solution for a value query $\psi$, we can run a binary search on the values, and for each value $\gamma$ to check whether $\psi[? \leftarrow \gamma]$ holds using \bfls model-checking. Thus, the number of calls to the \bfls model checker is logarithmic in $\sum_{e \in E}c(e)$, and hence polynomial in $|N|$. The PSPACE upper bound then follows from the PSPACE upper bound for \bfls model checking. The $\Delta_2^P$ upper bound for the network complexity follows from the $\Delta_2^P$ upper bound for the network complexity of \bfls model checking.
%orna4 what about the $\Delta_2^P$?
%gal4: added
\end{proof}

\section{Discussion}

We introduced the flow logic \bfls and studied its theoretical and practical aspects, as well as extensions and fragments of it. Below we discuss possible directions for future research.

At the more theoretical front, an important aspect that we left open in this work is the \emph{expressive power} of the different extensions and fragment of \bfls. As discussed in Remark~\ref{remark ep}, some cases follow easily from known expressiveness results for \ctls, yet the full picture is open. Also, as has been the case in traditional temporal logics, questions of \emph{succinctness} are  of interest too.

Another natural problem regarding \bfls is the \emph{satisfiability} problem. That is, given a \bfls formula, decide whether there is a flow network that satisfies it. In Section~\ref{unwinding}, we showed that, unlike \ctls, the logic \bfls is sensitive to unwinding. Thus, the standard algorithm for the satisfiability problem of \ctls, which is based on checking emptiness of tree automata, is not useful in the case of \bfls. Moreover, the satisfiability problem is challenging already to the linear fragment LFL of \bfls. Indeed, as discussed in Section~\ref{frags}, the
 semantics of LFL mixes linear and branching semantics, and there is no simple reduction of the satisfiability problem for LFL into the emptiness of word automata.

In the algorithmic side, the relation between maximum flows and minimum cuts has been significant in the context of the traditional maximum-flow problem. It would be interesting to study the relation between cuts and \bfls. In particular, the notion of cuts that satisfy some structural property that can be expressed with \bfls seems interesting.

Finally, many problems in various domains are solved using a reduction to the maximum-flow problem~\cite{CLR90,AMO93}. Flow logics allow reasoning about properties of flow networks that go beyond their maximal flow. An interesting direction for future research is to study applications of flow logics for rich variants of such problems. For example, using flow logics we can solve variants of matching or scheduling problems that involve restrictions on the allowed matches or variants of schedules.

\bibliographystyle{alpha}
\bibliography{ok}

\end{document}